\documentclass[11pt,letterpaper]{article}
\usepackage[margin=1in]{geometry}
\usepackage{latexsym,graphicx,amssymb,mathtools}
\usepackage{amsmath}
\usepackage{amsthm}
\usepackage{float}
\usepackage{xspace}
\usepackage{paralist}
\usepackage{enumerate,multicol}
\usepackage{cases}
\usepackage{caption}
\usepackage{graphicx}
\usepackage{hyperref}
\usepackage{tikz,pgfmath}
\usepackage{fourier}

\usepackage{xurl}
\usepackage{cleveref}
\usetikzlibrary{matrix,positioning,quotes}
\usetikzlibrary{shapes.multipart}
\usetikzlibrary{calc}
\usetikzlibrary{arrows,decorations.markings}
\usetikzlibrary{matrix}

\usepackage{tikz}
\usetikzlibrary{positioning}
\usetikzlibrary{decorations.text}
\usetikzlibrary{decorations.pathmorphing}
\usepackage[ruled,linesnumbered]{algorithm2e}

\usepackage[noend]{algpseudocode}
\usepackage{xcolor}
\usepackage{multirow}
\usepackage{multicol}
\usepackage{graphicx}
\usepackage{subcaption}
\usepackage{varwidth}
\usepackage{wrapfig}
\usepackage{bbm}
\usepackage{enumitem}
\usepackage{todonotes}

\DeclareMathOperator*{\E}{\mathbb{E}}
\DeclareMathOperator{\OPT}{\mathrm{OPT}}

\newcommand{\notshow}[1]{{}}
\newcommand{\AutoAdjust}[3]{{ \mathchoice{ \left#1 #2  \right#3}{#1 #2 #3}{#1 #2 #3}{#1 #2 #3} }}
\newcommand{\Xcomment}[1]{{}}

\newcommand{\InParentheses}[1]{\AutoAdjust{(}{#1}{)}}
\newcommand{\InBrackets}[1]{\AutoAdjust{[}{#1}{]}}
\newcommand{\InAngles}[1]{\AutoAdjust{\langle}{#1}{\rangle}}

\makeatletter
\newcommand{\vast}{\bBigg@{4}}
\newcommand{\Vast}{\bBigg@{5}}
\makeatother

\newcommand{\val}{\upsilon}
\newcommand{\uti}{\mu}
\newcommand{\utih}{\hat{\mu}}
\newcommand{\type}{t}
\newcommand{\dist}{D}

\newcommand{\bid}{{b}}
\newcommand{\balloc}{\boldsymbol{X}}
\newcommand{\alloc}{X}
\newcommand{\bpay}{\boldsymbol{p}}
\newcommand{\pay}{p}
\newcommand{\str}{{s}}

\newcommand{\indic}{\mathbbm{1}}

\newcommand{\half}{\frac{1}{2}}

\def\mymathhyphen{{\hbox{-}}}

\newcommand{\rev}{\mathrm{Rev}}

\newcommand{\efrev}{\mathrm{EF\mymathhyphen Rev}}
\newcommand{\rprev}{\mathrm{RPRev}}
\newcommand{\prev}{\mathrm{PRev}}

\newcommand{\core}{\textsc{Core}}
\newcommand{\coreh}{\widehat{\textsc{Core}}}
\newcommand{\tail}{\textsc{Tail}}
\newcommand{\single}{\textsc{Single}}

\newcommand{\M}{M}

\newcommand{\SFA}{\text{S1A}}
\newcommand{\SAP}{\text{SAP}}
\newcommand{\auc}{\mathcal{A}}
\newcommand{\SAprop}{$c$-efficient}
\newcommand{\SApropn}{$c$-efficiency}
\newcommand{\CEtuple}{$\left(\auc, \str,\dist, \{v_i\}_{i\in[n]}\right)$}
\newenvironment{prevproof}[2]{\noindent {\bf {Proof of {#1}~\ref{#2}:}}}{$\hfill \blacksquare$}

\newtheorem{theorem}{Theorem}[section]
\newtheorem{corollary}[theorem]{Corollary}
\newtheorem{lemma}[theorem]{Lemma}
\newtheorem{remark}{Remark}
\newtheorem{example}{Example}
\newtheorem{definition}{Definition}[section]

\title{Simultaneous Auctions are Approximately Revenue-Optimal for Subadditive Bidders}
\date{}
\author{Yang Cai\thanks{Yang Cai is supported by a Sloan Foundation Research Fellowship and the National Science Foundation Award CCF-1942583 (CAREER).} \\Yale University, USA\\yang.cai@yale.edu\\
\and Ziyun Chen\\Tsinghua University, China\\chenziyu20@mails.tsinghua.edu.cn\\
  \and Jinzhao Wu\thanks{Jinzhao Wu is supported by a Research Fellowship from the Center for Algorithms, Data, and Market Design at Yale (CADMY).}\\ Yale University, USA\\ jinzhao.wu@yale.edu}

\begin{document}

\maketitle

\begin{abstract}

We study revenue maximization in multi-item auctions, where bidders have subadditive valuations over independent items~\cite{rubinstein_simple_2015}. Providing a simple mechanism that is approximately revenue-optimal in this setting is a major open problem in mechanism design~\cite{cai_simple_2017}. In this paper, we present the first \emph{simple mechanism} whose revenue is at least a \emph{constant fraction} of the optimal revenue in multi-item auctions with subadditive bidders. 

Our mechanism is a simultaneous auction that incorporates either a personalized entry fee  or a personalized reserve price per item. We prove that for any simultaneous auction that satisfies \SApropn -- a new property we propose, its revenue is at least {an} $O(c)$-approximation to the optimal revenue. We further show that both the \emph{simultaneous first-price} and the \emph{simultaneous all-pay auction} are $1\over 2$-efficient. Providing revenue guarantees for non-truthful simple mechanisms, e.g., simultaneous auctions, in multi-dimensional environments has been recognized by Roughgarden et al.~\cite{roughgarden_price_2017} as an important open question. Prior to our result, the only such revenue guarantees are due to Daskalakis et al.~\cite{daskalakis_multi-item_2022} for bidders who have additive valuations over independent items.  Our result significantly extends the revenue guarantees of these non-truthful simple auctions to settings where bidders have {combinatorial} valuations.
\end{abstract}
     \thispagestyle{empty}
\addtocounter{page}{-1}
\newpage

\section{Introduction}
\label{sec:Intro}

Revenue-maximization in auctions is a central problem in both Economics and Computer Science due to its numerous applications in markets and online platforms. While Myerson's seminal work shows that a simple mechanism achieves the optimal revenue in single-item auctions~\cite{myerson_optimal_1981}, characterizing the revenue-optimal mechanism in multi-item settings has been notoriously difficult both analytically and algorithmically. Indeed, it has been shown that even finding (approximately) optimal multi-item mechanisms can require description complexity that is exponentially in the number of items, even for a single buyer~\cite{dughmi_sampling_2014,daskalakis_strong_2017,hart_selling_2019,babaioff_menu-size_2022}. Similarly, computing the revenue-optimal multi-item mechanism is known to be intractable even for basic settings~\cite{cai_reducing_2013,daskalakis_complexity_2014,chen_complexity_2015}. Furthermore, the revenue-optimal multi-item mechanisms may  exhibit several counter-intuitive properties that do not arise in single-item settings~\cite{briest_pricing_2010,hart_maximal_2015,hart_selling_2019}. To sum up, the optimal mechanism in multi-item settings is highly complex, difficult to characterize, and intractable to find.

Motivated by the highly complex nature of the optimal mechanism in multi-item settings, a recent line of work in algorithmic mechanism design~\cite{chawla_algorithmic_2007,chawla_multi-parameter_2010,alaei_bayesian_2011,hart_approximate_2012,kleinberg_matroid_2012,cai_simple_2013,babaioff_simple_2014,yao_n--1_2015,rubinstein_simple_2015,cai_duality-based_2021,chawla_mechanism_2016,cai_simple_2017, cai_simple_2019-2,dutting_olog_2020,cai_simple_2021, daskalakis_multi-item_2022,cai_computing_2022} investigate the inherent tradeoff between optimality and simplicity. In other words, \emph{can we use simple and practical mechanisms to approximate the optimal revenue in multi-item auctions}? The line of work mentioned above provide a positive answer in surprisingly general settings, under the standard item-independence assumption. In a beautiful work, D\"{u}tting et al.~\cite{dutting_olog_2020} show that a simple mechanism, known as sequential two-part tariff, can extract an $\Omega\left({1\over\log\log m}\right)$ fraction of the revenue when bidders have subadditive valuations, where $m$ is the number of items in the auction. A valuation $v:2^{[m]}\to \mathbb{R}_{\geq 0}$ is subadditive, if $v(S\cup T)\leq v(S)+v(T)$ for all sets of items $S,T\subseteq [m]$. Subadditivity captures the property that the items are not complements to each other, i.e., the items are not more valuable together than they are apart. This is a natural and important property in numerous economic environments. Hence, the following has been recognized as a fundamental open question:
\begin{align*}\label{eq:question1}
&\emph{Can we design simple mechanisms to achieve an \textbf{$O(1)$-approximation} to the optimal revenue}\\ &\quad\quad\quad\quad\quad\emph{when the bidders have \textbf{subadditive valuations} under the item-independence assumption?} \tag{*}
\end{align*}

Aside from question~\eqref{eq:question1}, other gaps remain in our understanding of the tradeoff between optimality and simplicity. In particular, existing results almost exclusively focus on truthful auctions, while many of the practical auctions are simple, but not truthful. For instance, the first-price auction is the most common type of mechanism in practice. In the display-ads market, arguably the most significant application of auctions in modern commerce, first-price auctions are adopted by every major exchange to allocate ad-displaying slots. Revenue guarantees for these simple non-truthful auctions have been scarce. Due to the ubiquity of such auctions, providing revenue guarantees for non-truthful simple mechanisms, especially in multi-item environments, has been recognized by Roughgarden et al.~\cite{roughgarden_price_2017} as an important open question:
\begin{align*}\label{eq:question2}
&\emph{Can we provide revenue guarantees for simple but \textbf{non-truthful} mechanisms in multi-item auctions}\\&\quad\quad\quad\quad\quad\quad\quad\quad\quad\emph{that match the guarantees for simple and truthful mechanisms?} \tag{**}
\end{align*}

\noindent Hartline et al.~\cite{hartline_price_2014} show that the first-price auction with reserve price (or minimum bid) achieves approximately optimal revenue in the single-item setting. Prior to our work, the only revenue guarantee for non-truthful auctions in multi-item settings is due to Daskalakis et al.~\cite{daskalakis_multi-item_2022}. They show that when the bidders have additive valuations, simultaneous auctions with entry fees or reserved prices can extract a constant fraction of the optimal revenue. 

We make significant progress in addressing both questions~\eqref{eq:question1} and~\eqref{eq:question2} in this paper. Our main result shows that the \emph{simultaneous first-price auction} (or the \emph{simultaneous all-pay auction}) with appropriately devised entry fees or reserve prices can achieve a constant fraction of the optimal revenue when bidders have subadditive valuations.

\subsection{Our Contributions}\label{sec:contribution}
We focus on the revenue guarantees of simultaneous auctions in this paper. We assume there are $n$ bidders and $m$ items. A simultaneous auction consists of $m$ parallel single-item auctions $\left\{\auc_j\right\}_{j\in [m]}$, one for each item. We consider two variants of simultaneous auctions:
\vspace{-.1in}
\begin{description}
    \item[Simultaneous auctions with personalized entry fees:] Each bidder $i$ is asked to pay a \emph{fixed entry fee} $\textsc{Ent}_i$ up front. The mechanism then proceeds to run the simultaneous auction, that is, run $m$ parallel single-item auctions. Only the bidders who pay the entry fees can participate in these single-item auctions. See Mechanism~\ref{alg:entryfee}~for details.\vspace{-.1in}

    \item[Simultaneous auctions with personalized reserve prices:] There is a reserve price $r_{ij}$ for each bidder $i$ and each item $j$. The mechanism runs the simultaneous auction. For each item $j$ that bidder $i$ wins, they need to pay the higher between their payment decided by the single-item auction $\auc_j$ and $r_{ij}$. See Mechanism~\ref{alg:reserveprice} for details.
\end{description}
\vspace{-.1in}
We now state our main result.

\medskip\begin{minipage}{0.94\textwidth}\textbf{Main Contribution:}
We identify a crucial property of simultaneous auctions $\auc=\left\{\auc_j\right\}_{j\in [m]}$ that we refer to as \SApropn, where $c$ is a positive real number (\Cref{def:good}). We show that, if the bidders have subadditive valuations over independent items (\Cref{def:subadditive}), for any~\SAprop~simultaneous auction $\auc$, there exists entry fees $\{\textsc{Ent}_i\}_{i\in [n]}$ and reserve prices $\{r_{ij}\}_{i\in [n],j\in[m]}$ such that the better of (i) $\auc$ with personalized entry fees $\{\textsc{Ent}_i\}_{i\in [n]}$ and (ii) $\auc$  with personalized reserve prices $\{r_{ij}\}_{i\in [n],j\in[m]}$ is an $O(c)$-approximation to the optimal revenue (\Cref{thm:main}). Next, we prove that both the simultaneous first-price auction and the simultaneous all-pay auction are $1\over 2$-efficient (\Cref{lem:firstprice,,lem:allpay} ). Hence, by incorporating with entry fees or reserve prices, the simultaneous first-price auction (or the simultaneous all-pay auction) is an $O(1)$-approximation to the optimal revenue (\Cref{cor:firstprice,,cor:allpay}). See \Cref{table:comparison} for comparison with other simple mechanisms.
\end{minipage}
\medskip

A few remarks are in order. Firstly, our benchmark is the optimal revenue achievable by any Bayesian Incentive Compatible mechanism (or equivalently achievable at any Bayes-Nash equilibrium of any mechanism, truthful or not). This is the standard benchmark considered in the simple vs. optimal literature and used in all~previous~results. Secondly, our result makes the standard item-independent assumption that is used in essentially all previous work regarding the tradeoff between simplicity and optimality in multi-item auctions for both truthful and non-truthful mechanisms~\cite{chawla_algorithmic_2007,chawla_multi-parameter_2010,alaei_bayesian_2011,hart_approximate_2012,kleinberg_matroid_2012,cai_simple_2013,babaioff_simple_2014,yao_n--1_2015,rubinstein_simple_2015,cai_duality-based_2021,chawla_mechanism_2016,cai_simple_2017, cai_simple_2019-2,dutting_olog_2020, daskalakis_multi-item_2022,cai_computing_2022}. 
{Without assuming item-independence, \cite{hart_selling_2019} and \cite{briest_pricing_2015} suggest that no mechanism with bounded menu complexity, a basic requirement for simple mechanisms, can offer any finite approximation guarantees, even when selling only two or three correlated items to a single buyer.~\footnote{{Our mechanism becomes either selling the grand bundle or selling the items separately when there is a single buyer, and hence has bounded menu complexity.}}}
 Finally, our approach fails to extend to simultaneous second-price auctions. We present some formal barriers in \Cref{ex:S2A}. See~\Cref{sec:mainthm} for a more detailed discussion. It is an interesting open question to understand whether some variant of the simultaneous second-price auction is approximately revenue-optimal in our setting.

\vspace{-.15in}
\paragraph{Revenue guarantees for a non-truthful auction.} {We provide details on how we evaluate the revenue of simultaneous auctions.} For the simultaneous auction with personalized reserved prices, our result holds even if the revenue is evaluated at the worst Bayes-Nash equilibrium. For the simultaneous auction with personalized entry fees, the answer is more {nuanced}. We show that for any Bayes-Nash equilibrium $\str$ of the original simultaneous auction, there exists a set of entry fees $\{\textsc{Ent}_i\}_{i\in[m]}$ such that (a) the set of Bayes-Nash equilibria remains unchanged 
in the new simultaneous auction with entry fees, 
and (b) our result holds for the revenue generated at {equilibrium} $\str$ in the new simultaneous auction with entry fees. Note that this is the same type of guarantee provided in~\cite{daskalakis_multi-item_2022} but for additive valuations. We believe such a guarantee is desirable in practice. When the original simultaneous auction has a unique Bayes-Nash equilibrium, our new mechanism inherits the uniqueness. When there are multiple equilibria, the auctioneer can first deploy the original simultaneous auction and wait until the bidders have reached an equilibrium $\str$. The auctioneer can now incorporate the set of entry fees {tailored for the equilibrium} $\str$. As our result suggests, the new mechanism 
 still admits $\str$ as a Bayes-Nash equilibrium and can now provide strong revenue guarantees. It seems unreasonable for the bidders to abandon $\str$ and play a different equilibrium in the new mechanism, while they choose to play according to $\str$ in the original one.

\paragraph{Our Techniques.}
Our result is based on a combination of the \SApropn~property for simultaneous auctions and the duality framework developed in~\cite{cai_duality-based_2021,cai_simple_2017}. Roughly speaking, a simultaneous auction is \SAprop, if for any Bayes-Nash equilibrium $\str$, any bidder $i$, and {\emph{any subset of items $S$}}, bidder $i$'s maximum attainable utility from items in $S$ plus the revenue generated from items in $S$ is at least $c$ times $i$'s value for the bundle $S$. It is not hard to see that if a simultaneous auction is \SAprop, then its welfare is at least $c$ times the optimal welfare. What we show is that this desirable property is also useful in producing revenue guarantees. 
Furthermore, we provide a simple but crucial modification for the double-core decomposition in the duality framework, which is  a most critical and challenging step of the entire analysis. This modification  allows us to extend the duality-based analysis to simultaneous auctions and will likely find further applications. 
With these two innovations, we avoid the type-loss tradeoff analysis, which is the major technical hurdle in~\cite{daskalakis_multi-item_2022}, and provide a modular and arguably simpler analysis for the significantly more general setting with subadditive bidders.

\begin{table}[H] 
\centering 
\vspace{-.1in}
\caption{A Summary of Approximation Results for Multi-Dimensional Revenue Maximization\\{\footnotesize S1A = Simultaneous First-Price Auction, S2A = Simultaneous-Second Price Auction, SAP = Simultaneous All-Pay Auction}}
\label{table:comparison}\begin{tabular}{|c|c|c|c|}
\hline
 &
  \begin{tabular}[c]{@{}c@{}}Sequential Two-Part\\ Tariff Mechanism\end{tabular} &
  \begin{tabular}[c]{@{}c@{}}S2A with \\ Entry Fees / Reserve Prices\end{tabular}&
  \begin{tabular}[c]{@{}c@{}}S1A, SAP with \\ Entry Fees / Reserve Prices\end{tabular}\\ \hline
\rule{0pt}{2.2ex} { Additive }    & $O(1)$\cite{chawla_mechanism_2016,cai_simple_2017}         &$O(1)$ {\cite{yao_n--1_2015,cai_duality-based_2021, daskalakis_multi-item_2022}}   & $O(1)$ for regular distributions\cite{daskalakis_multi-item_2022}                     \\ \hline
\rule{0pt}{2.2ex} { XOS }         & $O(1)$\cite{cai_simple_2017}         & $?$   & {\color{red} $O(1)$}(This paper)  \\ \hline
\rule{0pt}{2.2ex} { Subadditive } & $O(\log \log m)$\cite{dutting_olog_2020} & $?$  & {\color{red} $O(1)$}(This paper) \\ \hline
\end{tabular}
\end{table}
\vspace{-.15in}

{\paragraph{Approximate revenue monotonicity.} Building on our constant factor approximation, we establish approximate revenue monotonicity for subadditive bidders. This work generalizes the findings of Yao~\cite{yao_revenue_2017}, who demonstrate approximate revenue monotonicity for XOS bidders. The formal statement of the theorem and the accompanying proof can be found in Appendix~\ref{appendix:revmono}.}

\subsection{Additional Related Work}\label{sec:related work}

\paragraph{Simple vs. Optimal.} As we mentioned earlier, the majority of results in the simple vs. optimal literature focus on {truthful mechanisms}. Indeed, most of the designed mechanisms are \emph{dominant strategy incentive compatible}, providing very strong incentive guarantees for the bidders. However, to provide dominant strategy incentive compatibility, the mechanisms are sequential. As noted in~\cite{akbarpour_credible_2018}, the multi-round nature of these sequential mechanisms can present implementation difficulties that static mechanisms, such as simultaneous auctions, avoid.  Empirical evidence~\cite{athey_comparing_2011} also suggests that  static mechanisms can be conducted rapidly and asynchronously, thus offering several implementation benefits, which may explain the prevalence of static mechanisms in the real world.

\vspace{-.15in}
\paragraph{Algorithms for finding nearly revenue-optimal mechanisms.} There is a line of work focusing on efficient algorithms to find a $(1-\varepsilon)$-approximation of the optimal revenue in multi-item auctions~\cite{cai_optimal_2011,alaei_bayesian_2011,cai_optimal_2012,cai_algorithmic_2012,cai_reducing_2013,cai_understanding_2013,kothari_approximation_2019,cai_efficient_2021}. However, the computed mechanisms may not be simple, and might be too complicated to implement in practice.
\vspace{-.15in}
\paragraph{Welfare guarantees of simultaneous auctions.} A fruitful line of work aim to approximate the welfare in combinatorial auctions using simultaneous auctions. A non-exhaustive list includes~\cite{christodoulou_bayesian_2008,bhawalkar_welfare_2011,hassidim_non-price_2011,feldman_simultaneous_2013,dutting_valuation_2013,jose_correa_constant_2023}. Feldman et al. \cite{feldman_simultaneous_2013} show that, when bidders have subadditive valuations, the Price of Anarchy is $2$ for the simultaneous first-price auction, and $4$ for the simultaneous second-price auction under the no-overbidding assumption. Recently, Correa and Cristi~\cite{jose_correa_constant_2023} show that the Price of Anarchy is $6+\varepsilon$ for {a variant of the simultaneous all-pay auction.} We provide constant factor approximation to the optimal revenue using simultaneous auctions. Our analysis for the \SApropn~ property is inspired by~\cite{feldman_simultaneous_2013}.

\section{Preliminaries}
\label{sec:prelim}

In this paper, we focus on revenue maximization in simultaneous auctions with $n$ bidders and $m$ items. We represent the set of all $n$ bidders using $[n]$ and the set of all $m$ items with $[m]$. 

\vspace{-.15in}
\paragraph{Types and Valuation Functions.}

For each bidder $i$, its type $\type_i = \InAngles{\type_{ij}}_{j=1}^m$ is a $m$-dimensional vector where $\type_{ij}$ is the private information of bidder $i$ about item $j$. Each $\type_{ij}$ is drawn independently from the distribution $\dist_{ij}$. {The support of $\dist_i=\bigtimes_j \dist_{ij}$ and $\dist_{ij}$ are represented by $T_i$ and $T_{ij}$.} When bidder $i$ has a type $\type_i$, {their} valuation for a set of items $S$ is denoted as $\val_i(\type_i, S)$. {We  refer to $v_i(\cdot,\cdot)$ as bidder $i$'s valuation function that takes both $i$'s type and a set of items as input. We refer to $v_i(t_i,\cdot)$ as a valuation of bidder $i$, which only takes a set of items as input.}

{Throughout the paper, we assume that each bidder $i$'s distribution of valuation satisfies \Cref{def:subadditive}. This is colloquially referred to as bidder $i$'s valuation is \emph{subadditive over independent items}.} \Cref{def:subadditive} is proposed in~\cite{rubinstein_simple_2015} and has been adopted in essentially every work that studies revenue guarantees for simple mechanisms with subadditive bidders~\cite{cai_simple_2017, cai_learning_2017,dutting_olog_2020}.

\begin{definition}[Subadditive over independent items \cite{rubinstein_simple_2015}]
\label{def:subadditive}
A bidder $i$'s distribution $\mathcal{V}_i$ of their valuation  $\val_i\InParentheses{\type_i,\cdot}$ is subadditive over independent items if their type $\type_i$ is drawn from a product distribution $\dist_i = \bigtimes_j \dist_{ij}$ and $v_i(\cdot,\cdot)$ satisfies the following properties:
\begin{itemize}
    \item $\val_i\InParentheses{\cdot, \cdot}$ \textbf{has no externalities}. For each type $\type_i$ and any subset of items $S \subseteq [m]$, $\val_i(\type_i, S)$ relies solely on $\InAngles{\type_{ij}}_{j\in S}$. More formally, for any $\type_i, \type_i'$ such that $\type_{ij} = \type'_{ij}$ for all $j\in S$, $\val_i\InParentheses{\type_i, S} = \val_i\InParentheses{\type_i', S}$.
    \item $\val_i\InParentheses{\cdot,\cdot}$ \textbf{is monotone}. For any type $\type_i$ and $U\subseteq V \subseteq [m]$, $\val_i(\type_i, U) \leq \val_i(\type_i, V)$.
    \item $\val_i\InParentheses{\cdot,\cdot}$ \textbf{is subadditive}. For all $\type_i$ and $U, V\subseteq [m]$, $\val_i\InParentheses{\type_i, U\cup V} \leq \val_i\InParentheses{\type_i, U} + \val_i\InParentheses{\type_i, V}$.
\end{itemize}
    Similar to previous work, we use $V_i\InParentheses{\type_{ij}}$ to denote $\val_i\InParentheses{t_i, \AutoAdjust{\{}{j}{\}}}$ since it only depends on $\type_{ij}$.
\end{definition}

We provide an example in Appendix~\ref{appendix:example} to show how \Cref{def:subadditive} captures standard settings with independent items as special cases.

{An important property that we use in the analysis is the \emph{Lipschitzness} of the valuation function.  \begin{definition}\label{def:Lipschitz}
A valuation function $v(\cdot,\cdot)$ is \textbf{$\ell$-Lipschitz} if for any type $t,t'\in T$, and set $X,Y\subseteq [m]$,
$$\left|v(t,X)-v(t',Y)\right|\leq \ell\cdot \left(\left|X\Delta Y\right|+\left|\{j\in X\cap Y:t_j\not=t_j'\}\right|\right),$$ where $X\Delta Y=\left(X\backslash Y\right)\cup \left(Y\backslash X\right)$ is the symmetric difference between $X$ and $Y$.
\end{definition}}

\vspace{-.15in}
\paragraph{Combinatorial Auctions} We consider combinatorial auctions with $n$ bidders and 
$m$ items. In a combinatorial auction, each bidder observes their type $\type_i$ 
and chooses their action (e.g., a bid to submit) according to their type. We allow the bidders to use mixed strategies, that is, bidder $i$'s action $\bid_i$ is drawn from a distribution $\str_i(\type_i)$ that maps $i$'s type $\type_i$ to a distribution over possible actions. 
Given the action profile $\bid = \InParentheses{\bid_1, \bid_2,\cdots, \bid_n}$, the (possibly random) outcome of a combinatorial auction consists of a feasible allocation $\balloc\InParentheses{\bid} = (\alloc_1(\bid), \alloc_2(\bid),\cdots,  \alloc_n(\bid))\in 
\left({2^{[m]}}\right)^n$, where $\alloc_i(\bid)$ is set of items allocated to bidder $i$, and payments $\bpay(\bid) = (\pay_1(\bid), \pay_2(\bid),\cdots,\pay_n(\bid))$ for the bidders.
$u_i\InParentheses{t_i, b} = \E\InBrackets{\val_i\InParentheses{t_i, \alloc_i(\bid)} - \pay_i(\bid)}$ denotes the utility of bidder $i$ in the combinatorial auction when {their} type is $t_i$ under the action profile $\bid$.

\vspace{-.15in}

\paragraph{Simultaneous Auctions} 

A simultaneous auction consists of $m$ parallel single-item auctions $\left\{\auc_j\right\}_{j\in [m]}$. The action $\bid_i$ chosen by bidder $i$ is an $m$-dimensional vector in which the $j$-th coordinate $\bid_i^{(j)}$ represents the bid of bidder $i$ for item $j$. Let $\bid^{(j)} = \InParentheses{\bid_1^{(j)},\bid_2^{(j)}, \cdots, \bid_n^{(j)}}$ represent the collection of bids for item $j$. Each single-item auction $\auc_j$ runs independently to determine the allocation of item $j$ and each bidder's payment in $\auc_j$ according to $\bid^{(j)}$.  We use $\alloc_i^{(j)}(\bid^{(j)}) \subseteq \{j\}$ to denote the item that bidder $i$ gets and $\pay_i^{(j)}\InParentheses{\bid^{(j)}}$ to denote bidder $i$'s payment in the $j$-th auction. {Notice that $\alloc_i^{(j)}$ and $\pay_i^{(j)}$ might be random as the auction $\auc_j$ is allowed to be randomized.} 
 In a simultaneous auction, bidder $i$ receives all items won in each single-item auction $\auc_j$, i.e., $\alloc_i(\bid) =\bigcup_{j\in[m]}\alloc_i^{(j)}(\bid^{(j)})$, and their overall-payment $\pay_i(\bid) = \sum_{j\in[m]}\pay_i^{(j)}(\bid^{(j)})$ amounts to the sum of payments across the $m$ concurrent single-item auctions. {We also provide bidders with an additional action, denoted $\perp$, allowing them to abstain from bidding in a single-item auction.} {Bidding $\perp$ signifies that the bidder withdraws from competing for the item and incurs no payment for it.}

In this paper, we study two simultaneous auctions -- the \emph{simultaneous first-price auction (\SFA)} and the \emph{simultaneous all-pay auction (\SAP)}. {Both auctions} satisfy the \emph{highest bid wins} property, which states that, in each single-item auction, item $j$ is allocated to the bidder who submits the highest bid for $j$. In a \SFA, only the winning bidder for each item pays their bid; in a \SAP, all bidders pay their bids regardless of the outcome.

We formally define the notion of Bayes-Nash equilibrium in \Cref{appendix:example}. Let $\str$ be {a} Bayes-Nash equilibrium of auction $\auc$ w.r.t. distribution $\dist$, the expected revenue at equilibrium $\str$ is defined as
\[\rev^{(\str)}_{\dist}(\auc) = \sum_{i\in [n]}\E_{\type\sim \dist\atop \bid\sim \str(\type)}\InBrackets{\pay_i(\bid)}.\]

If $\auc$ is a simultaneous auction, we  use $\rev^{(\str)}_{\dist}(\auc, S)$ to denote the revenue of $\auc$ collected from items in $S$ at {equilibrium} $\str$:
\[\rev^{(\str)}_{\dist}(\auc, S) = \sum_{\substack{i\in [n]\\ j\in S}}\E_{\substack{\type\sim \dist\\\bid\sim \str(\type)}}\InBrackets{\pay_i^{(j)}(\bid)}\]
Finally, we define $\OPT(\dist)$ as the optimal revenue achievable by any randomized and Bayesian incentive compatible (BIC) mechanisms with respect to type distribution $D$ and valuation functions $\{v_i\}_{i\in[n]}$. Due to the revelation principle, we know that the highest revenue achievable by any auction at an Bayes-Nash equilibrium is {also} $\OPT(D)$.

\section{Our Mechanisms and Main Theorem}
\label{sec:mainthm&mecha}

\subsection{Our Mechanisms}
\label{subsec:ourmech}
We first introduce the two variations of simultaneous auctions that are used in our main theorem.

\vspace{-.15in}
\paragraph{{Simultaneous Auctions} with Entry Fees.}Our version of simultaneous auctions with entry fees is nearly identical to the one proposed by Daskalakis et al.~\cite{daskalakis_multi-item_2022}. For each bidder $i$, there is a personalized entry fee $e_i\in \mathbb{R}_{\geq 0}$, which does not depend on the bids submitted by the other bidders. Note that $e_i$ could depend on other parameters of the problem, e.g., the type distribution $\dist$, the valuation functions $\{v_i\}_{i\in[n]}$, and the equilibrium $\str$ { that we hope the bidders play}. The entry fee is charged with probability $1 - \delta$, and each bidder can decide whether to pay the entry fee to participate in the auction.

\begin{algorithm}[H]
\SetAlgorithmName{Mechanism}
~~\textbf{Input}: A simultaneous auction $\auc = (\alloc, \pay)$\, and $\{e_i\}_{i\in [n]} \in \mathbb{R}_{\geq 0}^n$\;
Each bidder $i$ submits a pair $\left(z_i, \bid_{i}\right)$ where $z_i \in \{0,1\}$ indicates whether bidder $i$ is willing to { accept an entry fee $e_i$ to enter the auction,} and $\bid_i$ is a $m$-dimensional vector representing  bidder $i$'s bid in $\auc$\;
{Independently for each bidder $i$,
the entry fee $\textsc{Ent}_i$ is set to $e_i$ with probability $1-\delta$ and is set of $0$ with probability $\delta$\;}
Run auction $\auc$ according to {the} bid profile $\bid = (\bid_1,\bid_2,\cdots,\bid_n)$\;
{Let $S = \{i: \textsc{Ent}_i = 0\text{ or } z_i = 1\}$} be the set of bidders that enters the auction, (i.e., bidders who agree to pay their entry fee)\;
{Each bidder $i\in S$ receives allocation $\alloc_i(b)$ and has payment $\pay_i(b)$. All other bidders receive nothing and pay nothing.}
\caption{{\sf  Simultaneous auction $\auc$ with personalized entry fee $\{e_i\}_{i\in [n]}$ \big($\auc^{(e)}_{\mathrm{EF}}$\big)}}\label{alg:entryfee}
\end{algorithm}

The probability that we do not charge the entry fee $\delta$ {should be thought of as a }very small positive constant. {In our proof, we choose $\delta$ to be $0.01$ and it suffices to guarantee \Cref{thm:main}.}

\vspace{-.15in}
\paragraph{{Simultaneous Auction} with Reserve Prices.}

The mechanism {first determines reserve prices $r_{ij}$ for each bidder $i$ and item $j$ using only information about the distribution of $V_i(t_{ij})$} (i.e., the distribution of bidder $i$'s value for winning only item $j$). As in standard simultaneous auction{s}, each bidder $i$ submits an $m$-dimensional bid vector $\bid_i$, where the $j$-th coordinate $b_i^{(j)}$ represents $i$'s bid for item $j$. 

Given the bid profile, the allocation is directly determined by the simultaneous auction $\auc$. If $i$ wins item $j$, $i$'s payment for item $j$ is the maximum of the reserve price $r_{ij}$ and $i$'s payment for item $j$ determined by {$\auc_j$}. {For the bidders who do not win item $j$, their payment for that item equals the payment determined by 
$\auc_j$.} The total payment of any bidder is the sum of their payments for all items.

\begin{algorithm}[H]
\SetAlgorithmName{Mechanism}
~~\textbf{Input}: A simultaneous auction $\auc = (\alloc, \pay)$ and a collection of reserved prices $\{r_{ij}\}_{i\in[n],j\in[m]}$\;

Each bidder $i$ submits their bid vector $\bid_{i}$, a $m$-dimensional vector, where $\bid_{i}^{(j)}$ can be $\perp$ for any $j$\;
Run auction $\auc$ with bid profile $\bid = (\bid_1,\bid_2,\cdots,\bid_n)$\;
Each bidder $i$ receives allocation $\alloc_i(\bid)$ and pays $\sum_{j\in \alloc_i(\bid)}\max\left\{p_i^{(j)}\left(\bid^{(j)}\right), r_{ij}\right\} + \sum_{\ell\notin \alloc_i(\bid)}p_i^{(\ell)}\left(b^{(\ell)}\right)$\;

\caption{{\sf $\auc$ with {personalized reserve prices} $\{r_{ij}\}_{i\in [n], j\in [m]}$ ($\auc^{(r)}_{\mathrm{RP}}$)}}\label{alg:reserveprice}
\end{algorithm}

\subsection{Main Theorem}
\label{sec:mainthm}

We introduce our main result in this section. We show that if a simultaneous auction $\auc$ satisfies certain desirable properties at a Bayes-Nash equilibrium $\str$, then the same auction $\auc$ that incorporates  additional entry fees or reserved prices can generate a constant fraction of the optimal revenue $\OPT(D)$ when bidders' valuations are \emph{subadditive over independent items}.

We first formally define the desirable properties :
\begin{definition}[{$c$-efficiency}]
\label{def:good}
Let $\str$ be a Bayes-Nash equilibrium of simultaneous auction $\auc$ w.r.t. type distribution $\dist$ and valuation functions $\{v_i\}_{i\in[n]}$. We define $\uti_i^{(\str)}(t_i, S)$ to be the optimal utility of bidder $i$ when their type is $t_i$, and they are only allowed to participate in the auctions for items in set $S$, while all other bidders bid according to $s_{-i}$. More specifically, 
\[\uti_i^{(\str)}\InParentheses{\type_i, S} = \sup_{{q_i\in \left(\mathbb{R}_{\geq 0}\cup \{\perp\}\right)^m}}\ \E_{\substack{\type_{-i} \sim \dist_{-i}\\ \bid_{-i}\sim \str_{-i}\InParentheses{\type_{-i}}}}\InBrackets{{\val_i\InParentheses{t_i,\alloc_i\InParentheses{q_i, \bid_{-i}} \cap S} -  \sum_{j\in S} \pay_i^{\left(j\right)}\InParentheses{q_i^{\left(j\right)}, \bid_{-i}^{\left(j\right)}}}}.\footnote{When $\alloc_i\InParentheses{q_i, \bid_{-i}}$ is a randomized allocation, $\alloc_i\InParentheses{q_i, \bid_{-i}}\cap S$ should be interpreted as only assigning $i$ the set of items in $W\cap S$, where $W$ is the random set of items that $i$ wins in $\alloc_i\InParentheses{q_i, \bid_{-i}}$.} \]

We say the tuple \CEtuple~is \SAprop~if the following conditions hold:

\begin{itemize}
    \item The payment for any item is non-negative. When a bidder bids {$\perp$} on {an} item, they pay nothing on this item regardless of the outcome.
 
    \item $\auc$ satisfies the highest bid wins property, i.e., for each item $j$, the bidder who has the highest bid wins item $j$.
    
    \item For any bidder $i$, any type $t_i$, and any set of items $S\subseteq [m]$,
            \[\uti_i^{(\str)}\InParentheses{t_i, S} + \rev^{(\str)}_{\dist}\InParentheses{\auc, S } \geq c\cdot \val_i\InParentheses{t_i, S }. \footnote{{Note that our definition is tailored for simultaneous auctions, as it is unclear how to define $\rev^{(\str)}_{\dist}\InParentheses{\auc, S }$ for general auctions. } }\]

\end{itemize}

\end{definition}

Before presenting our main theorem, we first discuss the {definition} of $c$-efficiency {and how it relates to several other important notions in mechanism design}. In \Cref{def:good}, the first and second conditions are easily satisfied by many simultaneous auctions, while the third condition is {crucial and more difficult to meet}. Indeed, any tuple \CEtuple~meeting the third condition implies that the equilibrium $\str$ achieves at least $c$ fraction of the optimal welfare. However, attaining a high welfare does not directly imply the third condition. We show that for the simultaneous second-price auction, there exists an instance $\left(\dist,\{\val_i\}_{i\in [n]}\right)$ with a no-overbidding equilibrium $\str$ {such that the third condition is violated} for any $c > 0$, but high welfare is still achieved at this equilibrium in the simultaneous second-price auction. See~\Cref{ex:S2A} for the complete construction.

{The third condition echoes the $(\lambda,\mu)$ smoothness condition introduced by Syrgkanis et al. \cite{syrgkanis_composable_2013}, albeit with three significant distinctions. First, our condition is specifically designed for simultaneous auctions and pertains to a particular Bayes-Nash equilibrium, in contrast to the  $(\lambda,\mu)$-smoothness which is generally applicable to any mechanism. Second, our condition imposes a lower bound on the utility of a single bidder, unlike the smoothness condition that considers the aggregate utility of all bidders. Lastly, our condition mandates the inequality to hold for every bundle $S$, a requirement absent in smooth mechanisms.}

{The third condition also notably aligns with the balanced prices framework~\cite{kleinberg_matroid_2012, feldman_combinatorial_2015, feldman_online_2016, dutting_olog_2020}, despite significant differences. Let $U$ be a set of items. The balanced prices framework assigns a price $p_i$ to each item $i \in U$ such that for any subset $S\subseteq U$, the buyer's utility from purchasing $S$ (i.e., $\val(S) - \sum_{i \in S} p_i$) combined with the revenue from the \emph{remaining set} (i.e., $\sum_{i\in U\backslash S} p_i$), approximates the total value of $U$. In contrast, our condition mandates that for any subset $S\subseteq U$, the buyer's utility, when bidding only on items in $S$ and acting in best response to other bidders' equilibrium strategies, along with the revenue from the \emph{same set} $S$, must attain a constant fraction of the total value of $U$. Additionally, while the balanced prices framework is limited to posted-price mechanisms, our definition can accommodate simultaneous auctions.}

Hartline et al. \cite{jason} introduce the concepts of competitive efficiency and individual efficiency for the single-dimensional setting. 
The third condition {in} \Cref{def:good} can be viewed as a generalization of these concepts in multi-dimensional settings. More specifically, in the single-item setting, for any mechanism that is $(\eta,\mu)$-individual and competitive efficient, our third condition holds for any equilibrium $\str$ with $c = \eta\mu$.

We now state our main theorem. 
\begin{theorem}
\label{thm:main}
{Let $\auc$ be a simultaneous auction, and $\str$ be a Bayes-Nash equilibrium of $\auc$ w.r.t. type distribution $\dist=\bigtimes_{i\in[n],j\in[m]}D_{ij}$ and valuation functions $\{v_i\}_{i\in[n]}$. If the distribution of bidder $i$'s valuation $v_i(t_i,\cdot)$ is subadditive over independent items (i.e., satisfies \Cref{def:subadditive})} and \CEtuple is \SAprop, then there exists a set of personalized entry fees $\{e_i\}_{i\in [n]}$ and a set of {personalized reserve prices} $\{r_{ij}\}_{i\in [n],j\in [m]}$ so that
\[\OPT(D) \leq \InParentheses{\frac{21}{c}\cdot \rev^{(\str)}_{\dist}\InParentheses{\auc^{(e)}_{\mathrm{EF}}} + \InParentheses{87 + \frac{51}{c}}\cdot \rev^{(\str')}_{\dist}\InParentheses{\auc_{\mathrm{RP}}^{(r)}}}{.}\]
{Here $\auc^{(e)}_{\mathrm{EF}}$ is auction $\auc$ with personalized entry fee $\{e_i\}_{i\in [n]}$. Note that $\auc_{\mathrm{EF}}^{(e)}$ has the same set of Bayes-Nash equilibria as $\auc$, so $s$ is also a Bayes-Nash equilibrium of $\auc_{\mathrm{EF}}^{(e)}$. $\auc^{(r)}_{\mathrm{RP}}$ is auction $\auc$ with reverse price $\{r_{ij}\}_{i\in [n], j\in [m]}$, and $\str'$ is an arbitrary Bayes-Nash equilibrium.}  

\end{theorem}

\begin{remark}
    Note that the entry fees $\{e_i\}_{i\in[n]}$ are selected based on $\str$. As stated in Lemma~\ref{lem:RandomSameEq}, a strategy profile $\str$ is a Bayes-Nash equilibrium in $\auc_{\mathrm{EF}}^{(e)}$ if and only if $\str$ is also a Bayes-Nash equilibrium in $\auc$. This implies that the introduction of entry fees does not give rise to any new equilibria, and { the same strategy profile} $\str$ continues to be an equilibrium. Therefore, it is reasonable to expect that the {bidders} to play according to the same equilibrium $\str$ after {introducing} the entry fees.
\end{remark}

See Section~\ref{sec:mecha} for {a detailed discussion about additional properties of equilibria in these two mechanisms}. Next, we argue that all equilibria of \SFA~and \SAP~are $1\over2$-efficient when bidders valuations are subadditive.

\begin{lemma}
\label{lem:firstprice}
 For any type distribution $D$, valuation functions $\{v_i\}_{i\in[n]}$, and any Bayes-Nash equilibrium $\str$ of \SFA, {as long as for any bidder $i$ and any $t_i$, $v_i(t_i,\cdot)$ is a subadditive function over $[m]$, $(\SFA, \str, \dist, \{v_i\}_{i\in[n]})$ is $\frac12$-efficient.}
\end{lemma}

\begin{lemma}
\label{lem:allpay}
For any type distribution $D$, valuation functions $\{v_i\}_{i\in[n]}$, and any Bayes-Nash equilibrium $\str$ of \SAP, {as long as for any bidder $i$ and any $t_i$, $v_i(t_i,\cdot)$ is a subadditive function over $[m]$, $(\SAP, \str, \dist, \{v_i\}_{i\in[n]})$ is $\frac12$-efficient.}
\end{lemma}

\begin{remark}
    Note that Lemma~\ref{lem:firstprice} and \ref{lem:allpay} do not require the bidders' valuations to be subadditive \emph{over independent items}. We only use item-independence in the proof of \Cref{thm:main}.
\end{remark}
The proofs of Lemma~\ref{lem:firstprice} and \ref{lem:allpay} are postponed to Appendix~\ref{appendix:proofofgood}. Combining Theorem~\ref{thm:main} with Lemma~\ref{lem:firstprice} and Lemma~\ref{lem:allpay}, we show that \SFA~and \SAP~with personalized entry fees or reserved prices can extract a constant fraction of the optimal revenue when the valuations are subadditive over independent~items.

\begin{corollary}
\label{cor:firstprice}
{For any type {distribution} $\dist=\bigtimes_{i\in[n],j\in[m]}D_{ij}$ and valuation functions $\{v_i\}_{i\in[n]}$, such that the distribution of bidder $i$'s valuation $v_i(t_i,\cdot)$ is subadditive over independent items (i.e., satisfies \Cref{def:subadditive}), if $\str$ is a Bayes-Nash equilibrium of the simultaneous first-price auction (\SFA),} then 
there exists a set of entry fees $\{e_i\}_{i\in [n]}$ and a set of reserve prices $\{r_{ij}\}_{i\in [n], j\in [m]}$ such that
\[\OPT(D) \leq 42\cdot \rev^{(\str)}_{\dist}\InParentheses{\mathrm{S1A}^{(e)}_{\mathrm{EF}}} + 189\cdot \rev^{(\str')}_{\dist}\InParentheses{\mathrm{S1A}_{\mathrm{RP}}^{(r)}},\]
{where $s$, a Bayes-Nash equilibrium of the original \SFA, remains to be a Bayes-Nash equilibrium for the \SFA~with personalized entry fees, and $s'$ is an arbitrary Bayes-Nash equilibrium of the \SFA~with reserve prices.}
\end{corollary}

\begin{corollary}
\label{cor:allpay}
{For any type $\dist=\bigtimes_{i\in[n],j\in[m]}D_{ij}$ and valuation functions $\{v_i\}_{i\in[n]}$, such that the distribution of bidder $i$'s valuation $v_i(t_i,\cdot)$ is subadditive over independent items (i.e., satisfies \Cref{def:subadditive}), if $\str$ is a Bayes-Nash equilibrium of the simultaneous all-pay auction (\SAP),} then there exists a set of entry fees $\{e_i\}_{i\in [n]}$ and a set of reserve prices $\{r_{ij}\}_{i\in [n], j\in [m]}$ such that
\[\OPT(D) \leq 42\cdot \rev^{(\str)}_{\dist}\InParentheses{\mathrm{SAP}^{(e)}_{\mathrm{EF}}} + 189\cdot \rev^{(\str')}_{\dist}\InParentheses{\mathrm{SAP}_{\mathrm{RP}}^{(r)}}{,}\]
{where $s$, a Bayes-Nash equilibrium of the original \SFA, remains to be a Bayes-Nash equilibrium for the \SFA~with personalized entry fees, and $s'$ is an arbitrary Bayes-Nash equilibrium of the \SFA~with reserve prices.}
\end{corollary}

\section{{Equilibria of Our Mechanisms}}
\label{sec:mecha}

In this section, we discuss some properties of the equilibrium in our mechanisms. {Note that a Bayes-Nash} equilibrium may not exist if the type spaces and action spaces are continuous. See Appendix~\ref{subsec:existence} for {a more detailed} discussion.

\subsection{Mechanisms with Entry Fees}
\label{subsec:entrymecha}

Notice that when the entry fee is charged deterministically, the bid vector $\bid_i$ has no impact on bidder $i$'s utility if they choose not to pay the entry fee. In this scenario, the bidder may report an arbitrary $\bid_i$, potentially introducing new {equilibria}. 
    As we show in \Cref{lem:RandomSameEq}, charging the entry fees randomly incentivizes each bidder to keep their bids even when they decide not to enter the auction. 
    Daskalakis et al.~\cite{daskalakis_multi-item_2022} provides an alternative mechanism with ``ghost bidders''. Their mechanism deterministically charges an entry fee and samples a bid from a "ghost bidder" in the execution of $\auc$ whenever a real bidder $i$ declines to pay the entry fee. As discussed in their paper, this mechanism is credible as the mechanism does not use any private randomness, but it may introduce new equilibria. We highlight that if we replace the randomized entry fees with deterministic ones together with ghost bidders, all claims in Theorem~\ref{thm:main} hold, except that {now} we need to evaluate the revenue of $\auc_{\mathrm{EF}}^{(e)}$ at a ``focal equilibrium'' that can be computed based on $\str$.

{Before examining the properties of $\auc_{\mathrm{EF}}^{(e)}$, it is essential to discuss a subtle detail concerning the actions of bidders in $\auc_{\mathrm{EF}}^{(e)}$. The actions available to bidder $i$ in $\auc_{\mathrm{EF}}^{(e)}$ has an additional dimension $z_i \in \{0,1\}$, that decides whether $i$ is willing to pay the entry fee. At any equilibrium $\str$,} it is clear that { bidder $i$} will choose to enter the auction if and only if \(\E_{\substack{\type_{-i}\sim \dist_{-i}\\\bid_{-i}\sim \str_{-i}(\type_{-i})}}\InBrackets{u_i(t_i, (b_i, b_{-i}))}\) exceeds $e_i$. {Therefore, $z_i$ depends exclusively on $b_i$ at any equilibrium. This allows for a liberal use of notation, interpreting the strategies of $\auc_{\mathrm{EF}}^{(e)}$ as a mapping from its type $\type_i$ to an $m$-dimensional bid vector $\bid_i$ (rather than to $(z_i,b_i)$).}

\begin{definition}[Strategy Profile of $\auc_{\mathrm{EF}}^{(e)}$ at Equilibrium $\str$]
\label{def:strprofile}
Suppose $\str$ is a Bayes-Nash equilibrium in auction $\auc_{\mathrm{EF}}^{(e)}$. For each bidder $i$,  its strategy profile $\str_i$ is defined as a mapping from type $\type_i$ to a distribution of {$m$}-dimensional bid {vectors}. Let 
\[u_i(t_i, b_i)  = \E_{\type_{-i}\sim \dist_{-i}}\InBrackets{\E_{\bid_{-i}\sim \str_{-i}(\type_{-i})}\InBrackets{u_i\InParentheses{\type_i,(\bid_i,\bid_{-i})}}}\] be the utility function for bidder $i$ in auction $\auc$ when their type is $\type_i$ and bids are $\bid_i$. When bidder $i$ participates in $\auc_{\mathrm{EF}}^{(e)}$ with type $\type_i$, she first samples a bid vector $\bid_i\sim \str_i\InParentheses{\type_i}$. Let $z_i = \indic\InBrackets{u_i(\type_i,b_i)\geq e_i}$ where $e_i$ is the entry fee for bidder $i$, she then submits $(z_i, b_i)$ as their action. It is clear that every equilibrium $\str$ of $\auc_{\mathrm{EF}}^{(e)}$ could be expressed in this form.
\end{definition}

\noindent The following lemma states that $\auc_{\mathrm{EF}}^{(e)}$ has exactly the same set of Bayes-Nash equilibria as $\auc$ for all $\delta \in (0, 1)$.

\begin{lemma}
\label{lem:RandomSameEq}
For any $\delta\in (0, 1)$, any set of entry fees $\{e_i\}_{i\in [n]}$,  any type distribution $\dist$, and valuation functions $\{v_i\}_{i\in[n]}$, a strategy profile $\str$ is a Bayes-Nash equilibrium in $\auc$ if and only if it is also a Bayes-Nash equilibrium in $\auc_{\mathrm{EF}}^{(e)}$.
\end{lemma}

\notshow{
\begin{remark}
Daskalakis et al.~\cite{daskalakis_multi-item_2022} provides an alternative mechanism with ``ghost bidders
''. This mechanism deterministically charges an entry fee and samples a bid from a "ghost bidder" in the execution of $\auc$ whenever a real bidder $i$ declines to pay the entry fee. As discussed in their paper, this mechanism is more practical but may introduce new equilibria. We point out that if we replace the randomized entry fees using the deterministic entry fees together with ghost bidders, all claims in Theorem~\ref{thm:main} \yangnote{still hold, except that we need to evaluate the revenue of $\auc_{\mathrm{EF}}^{(e)}$ at a ``focal equilibrium'' that can be computed using $s$.}
\end{remark}
}
We now discuss the revenue obtained by our mechanism with entry fees. The revenue consists of two parts: (i) the revenue derived from auction $\auc$, {i.e.,} $\rev_{\dist}^{(\str)}(\auc)$; (ii) the revenue obtained from the entry fees. We hereby provide a formal definition for the revenue generated from entry fees as follows.

\begin{definition}[Entry Fee Revenue]
\label{def:entryfee}
\[\efrev^{(\str)}_{\dist}(\auc) = \sup_{e\in \mathbb{R}^n_{\geq 0}} \sum_{i\in [n]} e_i\cdot \Pr_{\substack{\type_i\sim \dist_i}}\InBrackets{\E_{\substack{\type_{-i}\sim{\dist_{-i}}\\ \bid\sim\str(\type) }}\InBrackets{u_i(t_i, b)} \geq e_i}.\]
\end{definition}
It is important to note that the auction $\auc_{\mathrm{EF}}^{(e)}$ cannot fully obtain the revenue of auction $\auc$, i.e.,  $\rev_{\dist}^{(\str)}(\auc)$, and the revenue derived from entry fees, i.e., $\efrev^{(\str)}_{\dist}(\auc)$, at the same time. This is due to the fact that that when entry fees are imposed, bidders may refuse to enter the auction, which could potentially reduce the revenue generated by the auction $\auc$. Nevertheless, we could choose entry fees in a way to either maximize the revenue collected from the entry fees, thereby obtaining $\efrev^{(\str)}_{\dist}(\auc)$, or to set all entry fees to $0$ and attain $\rev_{\dist}^{(\str)}\InParentheses{\auc}$. In other words, $\rev^{(\str)}_{\dist}\InParentheses{\auc^{(e)}_{\mathrm{EF}}}$ is at least {$ \max\left\{\rev^{(\str)}_{\dist}(\auc), (1 - \delta - \varepsilon)\efrev^{(\str)}_{\dist}(\auc)\right\}$} for any $\varepsilon>0$.

\begin{lemma}
\label{lem:entryrev}
For any $\varepsilon > 0$, there exists a set of entry fees $\{e_i\}_{i\in [n]}$ so that
\[\rev^{(\str)}_{\dist}\InParentheses{\auc^{(e)}_{\mathrm{EF}}} \geq { \max\left\{\rev^{(\str)}_{\dist}(\auc), (1 - \delta - \varepsilon)\efrev^{(\str)}_{\dist}(\auc)\right\}}. \]
\end{lemma}

{The proofs of Lemma~\ref{lem:RandomSameEq} and Lemma~\ref{lem:entryrev} are postponed to Appendix~\ref{subsec:proofofRandomSameEq} and Appendix~\ref{subsec:proofofentryrev}, respectively.}

\subsection{Mechanisms with Reserve Prices}
\label{subsec:reservemecha}

{The following lemma provides a revenue guarantee for $\auc^{(r)}_{\mathrm{RP}}$.} 
{Importantly, this guarantee holds for any Bayes-Nash equilibrium of $\auc^{(r)}_{\mathrm{RP}}$.}

\begin{lemma}
\label{lem:reserveprice}
    For any type distribution $\dist$ and valuation functions $\left\{\val_i\right\}_{i\in[n]}$, if the simultaneous auction $\auc$ satisfies the first and second conditions of \Cref{def:good}, and \(\left\{r_{ij}\right\}_{i\in[n],j\in[m]}\) is a set of reserved prices that meets the following two conditions for some absolute constant \(b\in(0,1)\):
    \begin{itemize}
        \item[(1)] \(\sum_{i\in [n]}\Pr[V_{i}(t_{ij})\ge r_{ij}]\le b\), \(\forall j\in [m]\);
        \item[(2)] $\sum_{j\in[m]}\Pr[V_{i}(t_{ij})\ge r_{ij}]\le \half$, $\forall i\in [n]$,
    \end{itemize} then for any Bayes-Nash equilibrium $\str$ of the simultaneous auction with reserved prices $\auc_{\mathrm{RP}}^{(r)}$, the following revenue guarantee holds:
   \[\frac{2}{1-b}\cdot \rev^{(s)}_{D}\left(\auc^{(r)}_{\mathrm{RP}}\right) \geq \sum_{i,j}r_{ij}\cdot \Pr\left[V_i(t_{ij})\ge r_{ij}\right] .\]
\end{lemma}

The proof of Lemma~\ref{lem:reserveprice} is postponed to Appendix~\ref{subsec:proofofreserveprice}. 

\section{Proof of Theorem~\ref{thm:main}}
\label{sec:duality}

In this section, we complete the proof of Theorem~\ref{thm:main}. We extend the previous techniques, i.e., the duality framework \cite{cai_duality-based_2021,cai_simple_2017}, to simultaneous auctions by developing a new core-tail analysis. A crucial structure from the preceding approach hinged on the subadditivity and Lipschitzness of the bidders' utility functions. {Fortunately, the structure of simultaneous auctions ensures that the maximum utility a bidder can derive from a set of items (by bidding on them) remains a subadditive function. However, the Lipschitzness of the utility functions introduces additional subtlety. In simultaneous auctions, where bidding strategies form a Bayes-Nash equilibrium, each bidder faces  a distribution of prices, as opposed to a set of static prices, as encountered in posted price mechanisms analyzed in previous work.} This shift introduces a new challenge in controlling the Lipschitz constant of the utility functions, which, in turn, affects the concentration result.

We first introduce some notation. As in \cite{cai_simple_2017}, we assume that the type distributions are discrete{. See their paper for a discussion on how to convert continuous distributions to discrete ones without much revenue loss.} We fix type distribution $\dist$ in this section, and the probability {mass} functions of $\dist_i$ and $\dist_{ij}$ are denoted as $f_i(\cdot)$ and $f_{ij}(\cdot)$, respectively. Furthermore, the support of $\dist_i$ and $\dist_{ij}$ are represented by $T_i$ and $T_{ij}$. Recall that we define $V_i\InParentheses{\type_{ij}}$ as $\val_i\InParentheses{\type_i, \{j\}}$. We denote $F_{ij}$ as the distribution of $V_i\InParentheses{\type_{ij}}$ and let $\tilde{\varphi}_{ij}(x)$ be the Myerson's ironed virtual value~\cite{myerson_optimal_1981} of $x$ with respect to distribution $F_{ij}$. 

For any direct-revelation Bayesian Incentive Compatible mechanism $\M$, the allocation rule of $\M$ is represented by $\sigma$, wherein $\sigma_{iS}(\type_i)$ denotes the probability that bidder $i$ is allocated set $S$ with type $\type_i$. Given a set of parameters $\beta=\left\{\beta_{ij}\right\}_{i\in[n],j\in[m]} \in \mathbb{R}_{\geq 0}^{nm}$, we partition $T_i$ into $m + 1$ regions: (i) $R_{0}^{(\beta_i)}$ contains all types $\type_i$ satisfying $V_i(\type_{ij}) < \beta_{ij}$ for all $j\in [m]$. (ii) $R_j^{(\beta_i)}$ contains all types $t_i$ such that $V_i(\type_{ij}) \geq \beta_{ij}$ and $j$ is the smallest index in $\arg\max_k\left\{V_i(\type_{ik})-\beta_{ik}\right\}$. Intuitively, $R_j^{(\beta_i)}$ contains all types $\type_i$ for which item $j$ becomes the preferred item of bidder $i$ when the price for item $j$ is $\beta_{ij}$.

For each bidder $i$, define 
\[c_i = \inf\left\{x\geq 0:\sum_j\Pr_{\type_{ij}}\InBrackets{V_i\InParentheses{\type_{ij}}\geq \beta_{ij} + x} \leq \frac12\right\}.\]

For each $t_i\in T_i$, let $\mathcal{T}_i(\type_i) = \left\{j: V_i\InParentheses{\type_{ij}}\geq \beta_{ij} + c_i\right\}$ be the set of items that is above the price and $\mathcal{C}_i\InParentheses{\type_i} = [m]\backslash \mathcal{T}_i\InParentheses{\type_i}$ be its complement. Namely, if we set the reserve price (or posted price) of item $j$ for bidder $i$ at $\beta_{ij} + c_i$, it is very likely that bidder $i$ {will} buy at most one item. Thus, we could expect that {the contribution to revenue from $\mathcal{T}_i$ can be approximated by $\auc$ when incorporating reserve prices.} We now formally define the three components used to upper bound the optimal revenue below.

\begin{definition}
 For any feasible interim allocation {rule} $\sigma$ and any $\beta$, denote 
 \begin{align*}
     \single\InParentheses{\sigma, \beta} & = \sum_i \sum_{\type_i\in T_i} f_i(\type_i) \sum_{j\in [m]} \indic\InBrackets{\type_i\in R_j^{(\beta_i)}}\cdot \pi_{ij}(\type_i)\cdot \tilde{\varphi}_{ij}(\type_{ij}),\\
     \tail\InParentheses{\beta} &= \sum_i\sum_j\sum_{\type_{ij}:V_i(\type_{ij})\geq \beta_{ij} + c_i} f_{ij}(\type_{ij})\cdot V_i(\type_{ij})\sum_{k\neq j}\Pr_{\type_{ik}}\InBrackets{V_i\InParentheses{\type_{ik}} - \beta_{ik} \geq V_i\InParentheses{\type_{ij}}-\beta_{ij}},\\
     \core\InParentheses{\sigma,\beta} &= \sum_{i}\sum_{\type_i\in T_i}f_i(\type_i) \sum_{S\subseteq [m]} \sigma_{iS}\InParentheses{\type_i} \cdot \val_i\InParentheses{\type_i, S\cap \mathcal{C}_i\InParentheses{t_i}},
 \end{align*} where $\pi_{ij}\InParentheses{\type_i} = \sum_{S:j\in S}\sigma_{iS}\InParentheses{\type_i}$ is the probability that item $j$ is alloctaed to bidder $i$ with type $\type_i$.
\end{definition}

Let $\rev_{\dist}(\M)$ be the revenue of mechanism $\M$ while the bidders' types are drawn from the distribution $D$. Cai and Zhao~\cite{cai_simple_2017} show that the optimal revenue could be upper bounded by $\single$, $\tail$ and $\core$.

{
\begin{lemma}[\cite{cai_simple_2017}]
\label{lem:RevUpper}
For any BIC mechanism $\M$ 
and given any set of parameters $\beta=\{\beta_{ij}\}_{i\in[n],j\in[m]}\in \mathbb{R}^{nm}_{\geq 0}$, there exists a feasible interim allocation $\sigma^{(\beta)}$ so that 
\begin{align*}
\begin{split}
    \rev_{\dist}(\M) \leq 2\cdot \single\InParentheses{\sigma^{(\beta)}, \beta} + 4\cdot \tail\InParentheses{\beta} + 4\cdot \core\InParentheses{\sigma^{(\beta)}, \beta}.
\end{split}
\end{align*}

Additionally, for any constant $b\in (0,1)$ and any mechanism $\M$, there exists a set of parameters $\beta$ such that $\sigma^{(\beta)}$ satisfies the following two properties:
\begin{align}
     &\sum_{i} \Pr_{\type_{ij}}\InBrackets{V_i(\type_{ij}) \geq \beta_{ij}}\leq b, \quad \forall j\in [m]\label{eq:beta1} \\
     &\sum_{\type_i\in T_i} f_i(\type_i)\cdot \pi_{ij}^{(\beta)}\InParentheses{\type_i} \leq \Pr_{\type_{ij}}\InBrackets{V_i(\type_{ij})\geq \beta_{ij}} / b, \ \forall i\in [n],j\in [m],\text{ where }\pi^{(\beta)}_{ij}(\type_i) = \sum_{S:j\in S}\sigma_{iS}^{(\beta)}\InParentheses{\type_i}. \label{eq:beta2}
\end{align}
\end{lemma}
The first part of Lemma~\ref{lem:RevUpper}, namely the revenue guarantee, is derived by combining Theorem 2 and Lemma 14 from~\cite{cai_simple_2019} (the full version of~\cite{cai_simple_2017}), and hence, the proof is omitted here. In the second part of Lemma~\ref{lem:RevUpper}, we assert that parameters $\beta$ can be chosen such that the corresponding interim allocation $\sigma^{(\beta)}$
  satisfies two useful properties. This lemma is nearly identical to Lemma~5 in \cite{cai_simple_2019}, albeit with a minor alteration. It can be readily verified that the proof for Lemma 5 suffices to demonstrate this variation.
  }
\notshow{
\begin{lemma}[\cite{cai_simple_2017}]
\label{lem:RevUpper}
For any BIC mechanism $\M$ with interim allocation $\sigma$ and parameters $\beta=\{\beta_{ij}\}_{i\in[n],j\in[m]}\in \mathbb{R}^{nm}_{\geq 0}$, there exists another BIC mechanism $\M^{(\beta)}$ with interim allocation $\sigma^{(\beta)}$ so that 
\begin{align*}
\begin{split}
    \rev_{\dist}(\M) \leq 2\cdot \single\InParentheses{\sigma^{(\beta)}, \beta} + 4\cdot \tail\InParentheses{\beta} + 4\cdot \core\InParentheses{\sigma^{(\beta)}, \beta}.
\end{split}
\end{align*}
\end{lemma}

Lemma~\ref{lem:RevUpper} is simply a combination of Theorem~2 and Lemma~14 in~\cite{cai_simple_2019} (the full version of~\cite{cai_simple_2017}), and we omit the proof here. It is important to note that the parameters $\beta$ may be selected arbitrarily. In our following analysis, we choose parameters $\beta$ that satisfy the conditions stated in the subsequent lemma. This lemma is nearly identical to Lemma~5 in \cite{cai_simple_2019} with a minor adaption and its proof is postponed to Appendix~\ref{subsec:goodbetaappendix}.

\begin{lemma}
\label{lem:goodbeta}
For any constant $b\in (0,1)$ and any mechanism $\M$, there exists a set of parameters $\beta$ such that for the mechanism $\M^{(\beta)}$ constructed in Lemma~\ref{lem:RevUpper} by $\beta$, the following holds:
\begin{align}
     &\sum_{i} \Pr_{\type_{ij}}\InBrackets{V_i(\type_{ij}) \geq \beta_{ij}}\leq b, \quad \forall j\in [m]\label{eq:beta1} \\
     &\sum_{\type_i\in T_i} f_i(\type_i)\cdot \pi_{ij}^{(\beta)}\InParentheses{\type_i} \leq \Pr_{\type_{ij}}\InBrackets{V_i(\type_{ij})\geq \beta_{ij}} / b, \ \forall i\in [n],j\in [m],\text{ where }\pi^{(\beta)}_{ij}(\type_i) = \sum_{S:j\in S}\sigma_{iS}^{(\beta)}\InParentheses{\type_i}. \label{eq:beta2}
\end{align}
\end{lemma}
}

Suppose the simultaneous auction $\auc$ admits an equilibrium $\str$ under type distribution $\dist$ and valuation functions $\{\val_i\}_{i\in [n]}$ so that \CEtuple~is \SAprop~as defined in Definition~\ref{def:good}.  

We proceed to define the maximum revenue that can be achieved by simultaneous auction $\auc$ with reserve prices.
\begin{definition}
\label{def:rprev}
Define \(\rprev\) as the revenue obtainable by a simultaneous auction with optimal reserve prices $r_{ij}$'s, such that the revenue at its worst equilibrium $\str$ is maximized:

\[\rprev := \sup_{r} \inf_{\str\text{ is BNE}} \rev^{(\str)}_{\dist}\Big(\auc_{\mathrm{RP}}^{(r)}\Big) \]
\end{definition}

Given that \( \rprev \) is finite, the subsequent corollary directly follows.
\begin{lemma}
\label{coro:rprev}
For any $\varepsilon > 0$, there exists a set of reserve prices $\{r_{ij}\}_{i\in [n], j\in [m]}$ so that for any equilibrium $\str$ of $\auc_{\mathrm{RP}}^{(r)}$, its revenue at $\str$ achieves \((1 - \varepsilon)\rprev\).
\end{lemma}

In the following proof, we respectively approximate $\single, \tail$ and $\core$. Figure \ref{fig:proofroadmap} below offers a comprehensive overview of how we organize our proof.

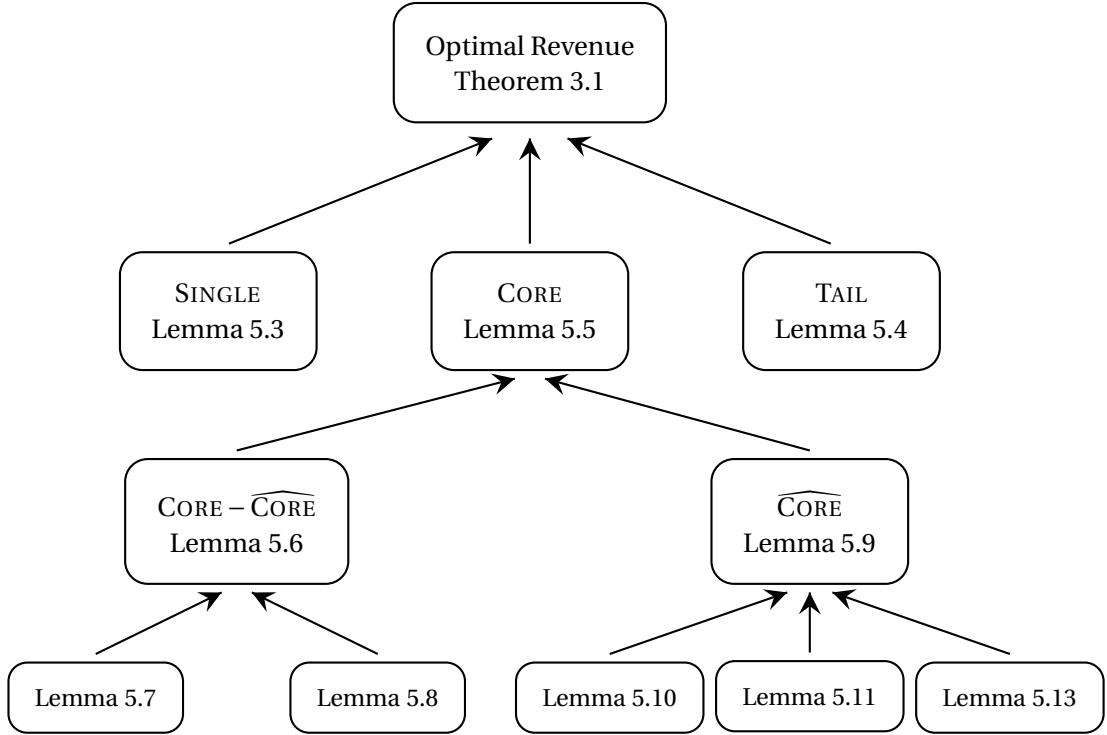
\begin{figure}[ht]
\centering
\begin{tikzpicture}[thick,amat/.style={matrix of nodes,nodes in empty cells,
  fsnode/.style={draw,solid,circle,execute at begin node={$u_{\the\pgfmatrixcurrentrow}$}},
  ssnode/.style={draw,solid,circle,execute at begin node={$v_{\the\pgfmatrixcurrentrow}$}}}]
  
 \node (ThmMain) [draw, rounded corners = 3mm,
                                 align=flush center, 
                                 inner sep=12 pt]
    {Optimal Revenue\\\Cref{thm:main}};

\coordinate[below = 0.2cm of ThmMain] (L1Mid) ;
\coordinate[left = 0.5cm of L1Mid] (L1Left) ;
\coordinate[right = 0.5cm of L1Mid] (L1Right) ;

\node (Core) [draw, rounded corners = 3mm,
                                 below = 1.5cm of L1Mid,    
                                 align=flush center, 
                                 inner sep=12 pt]
    {$\core$\\\Cref{lem:approxcore}};

\coordinate[above = 0.1cm of Core] (CoreLinkPos) ;

\node (Single) [draw, rounded corners = 3mm,
                                 left = 1.5cm of Core,    
                                 align=flush center, 
                                 inner sep=12 pt]
    {$\single$\\\Cref{lem:approxsingle}};

\coordinate[left = 4cm of CoreLinkPos] (SingleLinkPos) ;

\node (Tail) [draw, rounded corners = 3mm,
                                 right = 1.5cm of Core,    
                                 align=flush center, 
                                 inner sep=12 pt]
    {$\tail$\\\Cref{lem:approxtail}};

\coordinate[right = 4cm of CoreLinkPos] (TailLinkPos) ;

{\draw[decoration={markings,mark=at position 1 with
    {\arrow[scale=2,>=stealth]{>}}},postaction={decorate}] (SingleLinkPos) -- (L1Left) ;}

{\draw[decoration={markings,mark=at position 1 with
    {\arrow[scale=2,>=stealth]{>}}},postaction={decorate}] (CoreLinkPos) -- (L1Mid) ;}

{\draw[decoration={markings,mark=at position 1 with
    {\arrow[scale=2,>=stealth]{>}}},postaction={decorate}] (TailLinkPos) -- (L1Right) ;}

\coordinate[below = 2cm of Core] (Level2Mid) ;

\coordinate[below = 0.1cm of Core] (CoreDownLinkPosMid) ;

\coordinate[left = 0.2cm of CoreDownLinkPosMid] (CoreDownLinkLeft) ;

\coordinate[right = 0.2cm of CoreDownLinkPosMid] (CoreDownLinkRight) ;

\node (CoreH) [draw, rounded corners = 3mm,
                                 right = 2.4cm of Level2Mid,    
                                 align=flush center, 
                                 inner sep=12 pt]
    {$\coreh$\\\Cref{lem:coreh}};

\coordinate[above = 0.1cm of CoreH] (CoreHUpLinkPos) ;

{\draw[decoration={markings,mark=at position 1 with
    {\arrow[scale=2,>=stealth]{>}}},postaction={decorate}] (CoreHUpLinkPos) -- (CoreDownLinkRight) ;}

\node (CoreGap) [draw, rounded corners = 3mm,
                                 left = 2.4cm of Level2Mid,    
                                 align=flush center, 
                                 inner sep=12 pt]
    {$\core-\coreh$\\\Cref{lem:coregap}};
    
\coordinate[above = 0.1cm of CoreGap] (CoreGapUpLinkPos) ;

{\draw[decoration={markings,mark=at position 1 with
    {\arrow[scale=2,>=stealth]{>}}},postaction={decorate}] (CoreGapUpLinkPos) -- (CoreDownLinkLeft) ;}

\coordinate[below = 0.1cm of CoreGap] (CoreGapDownMid) ;
\coordinate[left = 0.2cm of CoreGapDownMid] (CoreGapDownLeft) ;
\coordinate[right = 0.2cm of CoreGapDownMid] (CoreGapDownRight) ;

\coordinate[below = 1.5cm of CoreGap] (Level3LeftMid) ;

\node (Tailtau) [draw, rounded corners = 2.5mm,
                                 left = 0.7cm of Level3LeftMid,    
                                 align=flush center, 
                                 inner sep=10 pt]
    {\small Lemma~\ref{lem:tailtau}};
    
\coordinate[above = 0.1cm of Tailtau] (TailtauPos) ;

{\draw[decoration={markings,mark=at position 1 with
    {\arrow[scale=2,>=stealth]{>}}},postaction={decorate}] (TailtauPos) -- (CoreGapDownLeft) ;}

\node (Sumofci)[draw, rounded corners = 2.5mm,
                                 right = 0.7cm of Level3LeftMid,    
                                 align=flush center, 
                                 inner sep=10 pt]
    {\small Lemma~\ref{lem:sumofci}};

\coordinate[above = 0.1cm of Sumofci] (SumofciPos) ;

{\draw[decoration={markings,mark=at position 1 with
    {\arrow[scale=2,>=stealth]{>}}},postaction={decorate}] (SumofciPos) -- (CoreGapDownRight) ;}

\coordinate[below = 0.1cm of CoreH] (CoreHDownMid) ;
\coordinate[left = 0.3cm of CoreHDownMid] (CoreHDownLeft) ;
\coordinate[right = 0.3cm of CoreHDownMid] (CoreHDownRight) ;

\coordinate[below = 1.5cm of CoreH] (Level3RightMid) ;

\node (utih)[draw, rounded corners = 2.5mm,
                                 left = 1.4cm of Level3RightMid,    
                                 align=flush center, 
                                 inner sep=10 pt]
    {\small Lemma~\ref{lem:utih}};

\node (EntryFee)[draw, rounded corners = 2.5mm,
                                 below = -0.5 cm of Level3RightMid,    
                                 align=flush center, 
                                 inner sep=10 pt]
    {\small Lemma~\ref{lem:entryfee}};

\node (SumofTau)[draw, rounded corners = 2.5mm,
                                 right = 1.4cm of Level3RightMid,    
                                 align=flush center, 
                                 inner sep=10 pt]
    {\small Lemma~\ref{lem:sumtau}};

\coordinate[above = 0.1cm of utih] (utihPos) ;
\coordinate[above = 0.1cm of EntryFee] (EntryFeePos) ;
\coordinate[above = 0.1cm of SumofTau] (SumofTauPos) ;

{\draw[decoration={markings,mark=at position 1 with
    {\arrow[scale=2,>=stealth]{>}}},postaction={decorate}] (utihPos) -- (CoreHDownLeft) ;}

{\draw[decoration={markings,mark=at position 1 with
    {\arrow[scale=2,>=stealth]{>}}},postaction={decorate}] (EntryFeePos) -- (CoreHDownMid) ;}

{\draw[decoration={markings,mark=at position 1 with
    {\arrow[scale=2,>=stealth]{>}}},postaction={decorate}] (SumofTauPos) -- (CoreHDownRight) ;}

\end{tikzpicture}
\caption{Relationships among lemmas and their roles in establishing the main result}
\label{fig:proofroadmap}
\end{figure}	

 {We first show that under  parameters $\beta$ and $\sigma^{(\beta)}$ that satisfy ~\eqref{eq:beta1} and~\eqref{eq:beta2} of \Cref{lem:RevUpper},} $\single$ and $\tail$ could be easily approximated by $\auc_{\mathrm{RP}}^{(r)}$ with appropriately selected reserve prices.

\begin{lemma}
\label{lem:approxsingle}    
For any $\sigma$ and $\beta$ that satisfy~\eqref{eq:beta1} and~\eqref{eq:beta2} {as stipulated in \Cref{lem:RevUpper},~\footnote{{Here it means that  $\sigma$ satisfies ~\eqref{eq:beta2} when $\sigma^{(\beta)}$ is replaced by $\sigma$.}}} $\single\InParentheses{\sigma,\beta}$ could be upper bounded by the revenue of a simultaneous auction with {personalized reserve prices }(Mechanism~\ref{alg:reserveprice}). That is to say, 
\[\single\InParentheses{\sigma, \beta} \leq 8 \cdot \rprev.\]
\end{lemma}

\begin{lemma}
\label{lem:approxtail}    
For any $\beta$ satisfying~\eqref{eq:beta1}, there exists a simultaneous auction with {personalized reserve prices} (Mechanism~\ref{alg:reserveprice}), whose revenue is at least $\frac{1-b}{2}\cdot\tail\InParentheses{\beta}$, i.e.,
\[\tail\InParentheses{\beta} \leq \frac{2}{1 - b} \cdot \rprev.\] 
\end{lemma}

The proofs of Lemma~\ref{lem:approxsingle} and Lemma~\ref{lem:approxtail} are {postponed to} Appendix~\ref{subsec:singleappendix} and \ref{subsec:tailappendix}. 

\subsection{The Analysis of the \(\core\)}\label{sec:core}

We now proceed to show that $\core$ could also be approximated by simultaneous auctions with entry fees or reserve prices.  
\begin{lemma}
\label{lem:approxcore}    For any $\sigma$ and $\beta$ that satisfy~\eqref{eq:beta1} and \eqref{eq:beta2} {as specified in \Cref{lem:RevUpper}},  {and tuple \CEtuple~that is \SAprop,}

\[\core(\sigma,\beta) \leq  \frac{4}{c}\cdot\efrev^{(\str)}_{\dist}\InParentheses{\auc} + \frac{1}{c}\cdot \rev^{(\str)}_{\dist}\InParentheses{\auc} + \InParentheses{\frac{2b + 2}{b(1-b)} + \frac{10}{c(1-b)}}\cdot \rprev,\] {where $\efrev^{(\str)}_{\dist}\InParentheses{\auc}$ is defined in \Cref{def:entryfee}.}
\end{lemma}

To prove Lemma~\ref{lem:approxcore}, we first introduce the double-core decomposition $\coreh$.

\begin{definition}[Double-core decomposition]
\label{def:tau}
    Let
    \[\tau_i := \inf\left\{ x\ge 0: \sum_j\Pr_{t_{ij}}\left[V_i(t_{ij})\ge \max\left\{\beta_{ij}, x\right\}\right]\leq \frac12 \right\}. \]
    and define $A_i$ to be {the set} $\left\{j: \beta_{ij}\le \tau_i\right\}$. Define $\coreh$ as 
    \[\coreh(\sigma,\beta) = \sum_i\sum_{t_i\in T_i}\sum_{S\subseteq [m]}f_i(t_i)\sigma_{iS}(t_i)\val_i\InParentheses{\type_i,S\cap Y_i(\type_i)} \]where $Y_i(t_i) = \{j: V_i(t_{ij}) < \tau_i\}$.
\end{definition}

\begin{remark}
\label{rem:corehdiffer}
We provide an alternative double-core decomposition compared to \cite{cai_simple_2017}. The main difference between these two decompositions is that the $\tau_i$ defined in our paper is different and could potentially be larger than theirs. As $\coreh$ defined in \cite{cai_simple_2017} is designed for posted-price mechanisms, they assign a price $Q_j$ for each item $j$ and replace $\max\left\{\beta_{ij},x\right\}$ by $\max\left\{\beta_{ij},x + Q_j\right\}$ in the definition of $\tau_i$. We show that the use of $Q_j$  is unnecessary. By Lemma~\ref{lem:sumtau}, {$\sum_i \tau_i$} in our paper can still be approximated by simple mechanisms. This is crucial for our analysis, as our proof highly replies on the $\tau_i$-{Lipschitzness} of $\uti_i$, and it can fail to be $\tau_i$-{Lipschitz} if we use $\tau_i$ defined in~\cite{cai_simple_2017}. 
\end{remark}
It suffices to demonstrate that our simultaneous auctions with either entry fees or reserve prices provide an upper bound for both the $\coreh$ and the difference between $\core$ and $\coreh$, i.e., $\core - \coreh$. We first show that the gap between these two cores could be approximated by the revenue of {simultaneous auction} with reserved prices.

\begin{lemma}
\label{lem:coregap}
For any $\sigma$, $\beta$ that satisfies~\eqref{eq:beta1} and \eqref{eq:beta2} {in \Cref{lem:RevUpper}},
    \[\core\InParentheses{\sigma,\beta} - \coreh\InParentheses{\sigma,\beta} \le \frac{2(b+1)}{b(1-b)}\cdot \rprev.\]
\end{lemma}

To prove Lemma~\ref{lem:coregap}, we first introduce {a} technical lemma that will be used in our proof.
\begin{lemma}
\label{lem:tailtau}
{For any $\beta$ that satisfies~\eqref{eq:beta1} in \Cref{lem:RevUpper}, }
\[ \sum_{i,j}\max\left\{\beta_{ij}, \tau_i\right\}\Pr_{t_{ij}}\left[V_i(t_{ij}) >\max\left\{\beta_{ij}, \tau_i\right\}\right] \le \frac{2}{1-b}\cdot \rprev\]
\end{lemma}
\begin{proof}

According to the definition of $\tau_i$, for every buyer $i$, $\sum_j\Pr_{t_{ij}}\left[V_i(t_{ij}) >\max\left\{\beta_{ij}, \tau_i\right\}\right]\le \half$. For each item $j$, {since} $\beta$ satisfies~\eqref{eq:beta1}, it holds that \[\sum_{i} \Pr_{t_{ij}}\left[V_i(t_{ij}) >\max\left\{\beta_{ij}, \tau_i\right\}\right] \leq \sum_{i} \Pr_{t_{ij}}\left[V_i(t_{ij}) >\beta_{ij} \right] \leq b.\] Applying lemma \ref{lem:reserveprice}, we then complete our proof.
\end{proof}

Before proving \Cref{lem:coregap}, we need one more lemma about $\sum_i c_i$.
\begin{lemma}
\label{lem:sumofci}

{For any $\beta$ that satisfies~\eqref{eq:beta1} in \Cref{lem:RevUpper}},
\[\sum_i c_i \le \frac{4}{1-b} \rprev.\]
\end{lemma}

\begin{proof}

Recall that \(c_i\) is defined as follows: \[c_i := \inf\left\{x\geq 0:\sum_j\Pr_{\type_{ij}}\InBrackets{V_i\InParentheses{\type_{ij}}\geq \beta_{ij} + x} \leq \frac12\right\}.\]

From the definition of \(c_i\), it directly follows that 
\[\sum_j\Pr_{\type_{ij}}\InBrackets{V_i\InParentheses{\type_{ij}}\geq \beta_{ij} + c_i} \leq \frac12\]  {for all} \(i\in [n]\). {Moreover, as the $\beta_{ij}$'s satisfy \eqref{eq:beta1}, the following is clearly true.
\[\sum_i\Pr_{\type_{ij}}\InBrackets{V_i\InParentheses{\type_{ij}}\geq \beta_{ij} + c_i} \leq \sum_i\Pr_{\type_{ij}}\InBrackets{V_i\InParentheses{\type_{ij}}\geq \beta_{ij}} \leq b, \ \forall j\in [m].\]}

Consequently, the set \(\left\{\beta_{ij}+ c_i\right\}_{i\in[n], j\in[m]}\) meets all conditions in Lemma~\ref{lem:reserveprice}. This leads to the implication that:
\[\sum_{i}\sum_{j}\InParentheses{\beta_{ij} + c_{i}} \cdot \Pr_{\type_i}\InBrackets{V_i\InParentheses{\type_{ij}} \ge \beta_{ij} + c_{i}} \le \frac{2}{1-b} \cdot \rprev.\]

On the other hand, 
\[\sum_{i}\sum_{j}\InParentheses{\beta_{ij} + c_{i}} \cdot \Pr_{\type_i}\InBrackets{V_i\InParentheses{\type_{ij}} \ge \beta_{ij} + c_{i}} \ge \sum_i\sum_j c_i \cdot \Pr_{\type_i}\InBrackets{V_i\InParentheses{\type_{ij}} \ge \beta_{ij} + c_{i}} \ge \half \cdot \sum_i c_i.\]
The last inequality arises since, when $c_i>0$, $\sum_j\Pr_{\type_i}\InBrackets{V_i\InParentheses{\type_{ij}} \ge \beta_{ij} + c_{i}} = \half$. Combining the two inequalities above, we know {that} $\sum_i c_i/2 \le \frac{2}{1-b} \cdot \rprev$.
\end{proof}

Now we are ready to prove {\Cref{lem:coregap}}.

\begin{proof}[Proof of Lemma~\ref{lem:coregap}]
 {Recall that 
    \begin{align*}
        \core\InParentheses{\sigma,\beta} &= \sum_{i}\sum_{\type_i\in T_i}f_i(\type_i) \sum_{S\subseteq [m]} \sigma_{iS}\InParentheses{\type_i}  \val_i\InParentheses{\type_i, S\cap \mathcal{C}_i\InParentheses{t_i}} \\
 \coreh(\sigma,\beta) &= \sum_i\sum_{t_i\in T_i}f_i(t_i)\sum_{S\subseteq [m]}\sigma_{iS}(t_i)\val_i\InParentheses{\type_i,S\cap Y_i(\type_i)}
    \end{align*}
    where \(\mathcal{C}_i \left(t_i\right) = \{j:V_i(\type_{ij}) < \beta_{ij} + c_i\}, Y_i(t_i) = \{j: V_i(t_{ij}) < \tau_i\}\).}
    
    Firstly, notice that
    \begin{align}
        \begin{split}
        \val_i\InParentheses{t_i,S\cap \mathcal{C}_i\InParentheses{t_i}} - \val_i\InParentheses{t_i,S\cap Y_i(t_i)}& \leq  \val_i\InParentheses{t_i, S\cap(\mathcal{C}_i(t_i) \backslash Y_i(t_i))}\\
        & \leq  \sum_{j\in S\cap\left(\mathcal{C}_i(t_i) \backslash Y_i(t_i)\right)} V_i\left(t_{ij}\right)\\
        & \leq \sum_{j\in S}V_i(t_{ij})\cdot \indic\left[\tau_i \le V_i(t_{ij}) \le \beta_{ij}+c_i\right]\\
         & \leq \sum_{j\in S} \InParentheses{\beta_{ij} \cdot \indic\InBrackets{V_i(\type_{ij}) \ge \tau_i} + c_i \cdot \indic\left[V_i(t_{ij}) \ge \max \left\{\beta_{ij},\tau_i\right\}\right]}\\
        \end{split}
    \end{align}
    {The last} inequality is because when $\tau_i \le V_i(t_{ij}) \le \beta_{ij}+c_i$, $V_i(t_{ij})$ is upper bounded by $\beta_{ij}$ when $V_i(t_{ij})\le \beta_{ij}$ and upper bounded by $\beta_{ij}+c_i$ when $V_i(t_{ij})\ge \beta_{ij}$. Hence
    \begin{align}\label{eq:coregap1}
        \begin{split}
        & \core - \coreh\\
        & \le  \sum_i\sum_{t_i}\sum_{S\subseteq [m]} \sum_{j\in S} f_i(t_i)\sigma_{iS}(t_i) \cdot \left(\beta_{ij} \cdot \indic[V_i(t_{ij}) \ge \tau_i] + c_i \cdot \indic\left[V_i(t_{ij}) \ge \max \{\beta_{ij},\tau_i\}\right]\right)\\
        & =  \sum_i\sum_j\sum_{t_i} f_i(t_i)\cdot \pi_{ij}(t_i)\cdot \left(\beta_{ij} \cdot \indic\left[V_i(t_{ij}) \ge \tau_i\right] + c_i \cdot \indic\left[V_i(t_{ij}) \ge \max \left\{\beta_{ij},\tau_i\right\}\right]\right).\\
        \end{split}
    \end{align}

    First we bound $\sum_i\sum_j\sum_{t_i}f_i(t_i)\cdot \pi_{ij}(t_i)\cdot \beta_{ij} \cdot \indic\left[V_i(t_{ij}) \ge \tau_i\right]$.
    \begin{align}\label{eq:coregap2}
        \begin{split}
        & \sum_i\sum_j\sum_{t_i}f_i(t_i)\pi_{ij}(t_i)\cdot \beta_{ij} \cdot \indic\left[V_i(t_{ij}) \ge \tau_i\right]\\
        \le & {\sum_i}\sum_{j\in A_i} \beta_{ij}\cdot \sum_{t_i}f_i(t_i)\cdot \indic\left[V_i(t_{ij})\ge \tau_i\right] + \sum_i\sum_{j\notin A_i}\beta_{ij}\cdot \sum_{t_i}f_i(t_i)\pi_{ij}(t_i)\\
        \le & {\sum_i} \sum_{j\in A_i}\beta_{ij}\cdot \Pr_{t_{ij}}\left[V_i(t_{ij})\ge \tau_i\right] + \sum_i\sum_{j\notin A_i}\beta_{ij}\cdot \Pr_{t_{ij}}\left[V_i(t_{ij})\ge \beta_{ij}\right]/b\\
        \le & \frac{1}{b} \cdot \sum_{i,j}\max\left\{\beta_{ij},\tau_i\right\} \cdot \Pr_{t_{ij}}\left[V_i(t_{ij}) \ge \max\left\{\beta_{ij},\tau_i\right\}\right]\\
        \le & \frac{2}{b(1-b)}\cdot \rprev\\
        \end{split}
    \end{align} {The set \(A_i\) is defined as \(\{j:\beta_{ij} \leq \tau_i\}\)} in \Cref{def:tau}. {The parameters \( \beta_{ij} \)'s satisfy Inequality~\eqref{eq:beta2}, as presented in the statement of the lemma, which substantiates the second inequality.}
 The third inequality is due to the definition of \(A_i\) and the last inequality follows from lemma \ref{lem:tailtau}. 

   {Secondly,} we bound $\sum_i\sum_j\sum_{t_i}f_i(t_i)\pi_{ij}(t_i)\cdot c_i \cdot \indic\left[V_i(t_{ij}) \ge \max\left\{\beta_{ij}, \tau_i\right\}\right]$.
    \begin{align}\label{eq:coregap3}
        \begin{split}
        & \sum_i\sum_j\sum_{t_i}f_i(t_i)\pi_{ij}(t_i)\cdot c_i \cdot \indic\left[V_i(t_{ij}) \ge \max\left\{\beta_{ij}, \tau_i\right\}\right]\\
        \le & \sum_i c_i \sum_j\sum_{t_i}f_i(t_i)\cdot \indic\left[V_i(t_{ij}) \ge \max\left\{\beta_{ij}, \tau_i\right\}\right]\\
        = & \sum_i c_i \sum_j \Pr\left[V_{ij}(t_i) \ge \max\left\{\beta_{ij}, \tau_i\right\}\right]\\
        \le & \sum_i c_i/2\\
        \le & \frac{2}{1-b}\cdot \rprev\\
        \end{split}
    \end{align} where the second inequality {is due to} the definition of $\tau_i$ \big(\Cref{def:tau}\big) and the last inequality is due to {\Cref{lem:sumofci}}.    Combining~\eqref{eq:coregap1}, \eqref{eq:coregap2} and \eqref{eq:coregap3}, we complete our proof.
\end{proof}

Next, we argue that $\coreh$ could be approximated by auction $\auc$ with either entry fees or reserve prices. 
\begin{lemma}
\label{lem:coreh} For any $\sigma$ and $\beta$ that satisfy~\eqref{eq:beta1} and \eqref{eq:beta2} {in \Cref{lem:RevUpper}}, {and tuple \CEtuple~that is \SAprop,}, it holds that
\[\coreh\InParentheses{\sigma, \beta} \le \frac{1}{c}\InParentheses{4\cdot\efrev^{(\str)}_{\dist}\InParentheses{\auc} + \rev^{(\str)}_{\dist}\InParentheses{\auc} + \frac{10}{1-b}\cdot \rprev},\]  where $\efrev^{(\str)}_{\dist}\InParentheses{\auc}$ denotes the revenue derived from entry fees, as defined in \Cref{def:entryfee}.
\end{lemma}

Recall that in Definition~\ref{def:good}, we define $\uti_i^{(\str)}(t_i, S)$ as the optimal utility that  bidder $i$ can attain when only {the bundle} $S$ is {available}. We further define $\utih_i\InParentheses{\type_i,S}$ as $\uti_i^{(\str)}\InParentheses{t_i, S\cap Y_i(\type_i)}$ where $Y_i(\type_i) = \left\{j: V_i\InParentheses{\type_{ij}} < \tau_i\right\}$.  Lemma~\ref{lem:subadditiveofmu} demonstrates that $\utih\InParentheses{t_i,\cdot }$ satisfies monotonicity, subadditivity, no externalities and $\tau_i$-Lipschitzness. {Our proof of Lemma~\ref{lem:coreh} can be divided into the following three steps. The first step, summarized in \Cref{lem:utih}, argues that the ``truncated'' utility, represented as $\sum_i \E_{\type_i \sim \dist_i}\InBrackets{\utih_i\InParentheses{\type_i,[m]}}$,  together with the revenue of the auction \(\auc\) serves as a $c$-approximation to \(\coreh\) by employing the third property in the definition of \SApropn. The second step, i.e., \Cref{lem:entryfee}, shows how to extract revenue from the ``truncated'' utility by setting a entry fee at the median of the utility function. We demonstrate that the corresponding revenue is high enough using a concentration inequality for subadditive functions.  The last step, i.e.,~\Cref{lem:sumtau} shows that the difference between the revenue from the entry fees and the truncated utilities can be approximated by the revenue from another simultaneous auction with reserved prices.}

\begin{lemma}
\label{lem:utih}
{For any $\sigma$, $\beta$ that satisfies~\eqref{eq:beta1} and \eqref{eq:beta2} in \Cref{lem:RevUpper},}
\[\sum_i \E_{\type_i \sim \dist_i}\InBrackets{\utih_i\InParentheses{\type_i,[m]}} \ge c\cdot \coreh\InParentheses{\sigma, \beta} - \rev_{\dist}^{(\str)}\InParentheses{\auc}.\]
\end{lemma}
\begin{proof}

The third property of \SApropn~\big(\Cref{def:good}\big) states that for any $S\subseteq [m]$, 
\[\uti_i^{\InParentheses{\str}}\InParentheses{\type_i,S\cap Y_i \InParentheses{\type_i}} \geq c\cdot \val_i\InParentheses{\type_i, S\cap Y_i\InParentheses{\type_i}} - {\rev^{(\str)}_D\InParentheses{\auc,S\cap Y_i\InParentheses{\type_i}}}.\]

By the definition of \(\utih_i\) and the monotonicity of \(\uti_i^{\InParentheses{\str}}\InParentheses{\type_i,\cdot}\), it follows that
\begin{align}
\label{eq:utih1}
\begin{split}
    &\sum_i \E_{\type_i \sim \dist_i}\InBrackets{\utih_i\InParentheses{\type_i,[m]}} \\
    & \ge \sum_i \E_{\type_i \sim \dist_i} \InBrackets{\sum_{S\subseteq [m]} \sigma_{iS} \InParentheses{\type_i}\uti_i^{\InParentheses{\str}}\InParentheses{\type_i,S\cap Y_i \InParentheses{\type_i}}}\\
    & \ge \sum_i \E_{\type_i \sim \dist_i}\InBrackets{\sum_{S\subseteq [m]} \sigma_{iS}\InParentheses{\type_i} \InParentheses{c\cdot \val_i\InParentheses{\type_i, S\cap Y_i\InParentheses{\type_i}} - {\rev^{(\str)}_D\InParentheses{\auc,S\cap Y_i\InParentheses{\type_i}}}}}\\
    & = c\cdot \sum_i\sum_{\type_i \in T_i}\sum_{S\subseteq [m]} f_i\InParentheses{\type_i}\sigma_{iS}\InParentheses{\type_i}\val_i\InParentheses{\type_i, S\cap Y_i\InParentheses{\type_i}} - \sum_i\sum_{\type_i \in T_i}\sum_{S\subseteq [m]}f_i\InParentheses{\type_i}\sigma_{iS}\InParentheses{\type_i}\rev_{\dist}^{(\str)}\InParentheses{{\auc,}S \cap Y_i(\type_i)}\\
\end{split}
\end{align}

The first term here is { exactly $c\cdot \coreh$. Recall the definition of $\coreh$:}
\begin{align}\label{eq:utih2}
    \begin{split}
       \coreh(\sigma,\beta) = \sum_i\sum_{t_i\in T_i}\sum_{S\subseteq [m]}f_i(t_i)\sigma_{iS}(t_i)\val_i\InParentheses{\type_i,S\cap Y_i(\type_i)}.
    \end{split}
\end{align}

We are only left to upper bound the second term. Recall that $\pi_{ij}\left(\type_i\right)$ represents the probability that item $j$ is allocated to bidder $i$, meaning that \(\sum_i\sum_{t_i}f_i\InParentheses{\type_i}\pi_{ij}(t_i) \le 1\) for all $j\in[m]$. Consequently, 
\begin{align}\label{eq:utih3}
 \begin{split}
   \sum_i\sum_{\type_i \in T_i}\sum_{S\subseteq [m]}f_i\InParentheses{\type_i}\sigma_{iS}\InParentheses{\type_i}\rev_{\dist}^{(\str)}\InParentheses{{\auc,}S \cap Y_i(\type_i)} & \leq \sum_i\sum_{\type_i \in T_i}\sum_{S\subseteq [m]}f_i\InParentheses{\type_i}\sigma_{iS}\InParentheses{\type_i}\rev_{\dist}^{(\str)}\InParentheses{{\auc,}S}\\
   & = \sum_j \rev_{\dist}^{(\str)}\InParentheses{\auc,\left\{j\right\}}\sum_i\sum_{t_i}f_i\InParentheses{\type_i}\sum_{S:j\in S}\sigma_{iS}\InParentheses{S}\\
    & = \sum_j \rev_{\dist}^{(\str)}\InParentheses{\auc,\left\{j\right\}}\sum_i\sum_{t_i}f_i\InParentheses{\type_i}\pi_{ij}\InParentheses{\type_i}\\
    & \le  \rev_{\dist}^{(\str)}\InParentheses{\auc, [m]}.
    \end{split}
\end{align}
{The} first inequality employs the monotonicity of $\rev_{\dist}^{(\str)}$, and the first equation is because that $\auc$ is a simultaneous auction, thereby making its revenue additive across items.

Putting~\eqref{eq:utih1},\eqref{eq:utih2} and \eqref{eq:utih3} together, we then finish our proof.

\end{proof}

\notshow{
{\color{blue}
\begin{remark}
\label{rem:poa}
This lemma states that the truncated utility plus the original revenue is a constant approximation to the  truncated welfare when bidders are playing the original Bayes-Nash equilibrium. Although simultaneous second price auction has a $1/4$ PoA bound w.r.t. the welfare~\cite{feldman_simultaneous_2013}, it does not imply the lemma above
\end{remark}}
}

Finally, notice that $\utih_i(\cdot, \cdot)$ is a subadditive function that is $\tau_i$-Lipschitz. To approximate $\sum_i \E_{\type_i \sim \dist_i}\InBrackets{\utih_i\InParentheses{\type_i,[m]}}$, the concentration {inequality} for subadditive functions {tells} us that we can extract the revenue from the bidder's utility by setting {an entry fee} at its median.

\begin{lemma}
\label{lem:entryfee}
There exists bidder-specific entry-fees $\{e_i\}_{i\in [n]}$, such that
\[\sum_i \E_{\type_i \sim \dist_i}\InBrackets{\utih_i\InParentheses{\type_i,[m]}} \le 4\cdot\efrev^{(\str)}_{\dist}\InParentheses{\auc} + \frac52\sum_i\tau_i.\]
\end{lemma}
\begin{proof}
We first introduce a concentration inequality for subadditive function from Corollary~1 in \cite{cai_simple_2019}.

\begin{lemma}[\cite{cai_simple_2017}]
\label{lem:concentrate}
    Let $g(t, \cdot)$ with $t\sim \dist = \bigtimes_j \dist_j$ be a function drawn from a distribution that is subadditive over independent items of ground set $I$. Assume that the function \(g( \cdot, \cdot )\) exhibits \(c\)-Lipschitzness. Let \(a\) represent the median of {the random variable} \(g(t, I)\), that is, \(a = \inf \left\{ x \geq 0 : {\Pr_t[g(t, I) \leq x]} \geq \frac{1}{2} \right\}\).{Therefore,}\[\E_t[g(t,I)]\le 2a + \frac{5c}{2}.\]
\end{lemma}
Notice that $\utih(\type_i, [m])$ is a random variable in which the randomness comes from its random type $\type_i$. Let $e_i$ be the median of $\utih_i(t_i, [m])$. Since $\utih\InParentheses{\cdot, \cdot}$ is subadditive over independent items and  $\tau_i$-Lipschitz {by Lemma~\ref{lem:subadditiveofmu}}, {Lemma~\ref{lem:concentrate} implies the following}
\begin{equation}\label{eq:entryfee1}
\E_{\type_i \sim \dist_i}\InBrackets{\utih\InParentheses{\type_i, [m]}}\le 2e_i + \frac52  \tau_i.    
\end{equation}

The monotonicity of $\uti_i$ implies that $\uti_i\InParentheses{t_i, [m]} \ge \utih_i\InParentheses{t_i, [m]}$. Therefore, if we set the entry fee as $e_i$, i.e.,  the median of $\utih_i(t_i, [m])$, the probability that bidder $i$ pays the entry fee is at least $1/2$. Thus
\begin{equation}\label{eq:entryfee2}
\efrev^{(\str)}_{\dist}\InParentheses{\auc} \geq \sum_i e_i\Pr_{\type_i \sim \dist_i}\left[\uti_i\InParentheses{\type_i, [m]} \ge e_i\right] \ge \half \sum_i e_i.
\end{equation}

Combining \eqref{eq:entryfee1} and \eqref{eq:entryfee2},we then get
\begin{align*}
    \sum_i \E_{\type_i \sim \dist_i}\InBrackets{\utih_i\InParentheses{\type_i,[m]}} & \le 2\sum_i e_i + \frac52\sum_i\tau_i\\
    & \le 4\cdot\efrev^{(\str)}_{\dist}\InParentheses{\auc} + \frac52\sum_i\tau_i.
\end{align*}
\end{proof}

As the last step, we show that the sum of the Lipschitz constant $\sum_{i}\tau_i$ can be approximated by  $\rprev$.

\begin{lemma}
\label{lem:sumtau}
{For any $\beta$ that satisfies~\eqref{eq:beta1} in \Cref{lem:RevUpper}, }
    \[\sum_i\tau_i \le \frac{4}{1-b}\cdot \rprev.\]
\end{lemma}
\begin{proof}
Notice that

\begin{equation}
\label{eq:tau2}
    \sum_{i,j}\max\left\{\beta_{ij}, \tau_i\right\}\Pr_{t_{ij}}\left[V_i(t_{ij}) >\max\left\{\beta_{ij}, \tau_i\right\}\right]\ge \sum_{i,j} \tau_i\Pr_{t_{ij}}\left[V_i(t_{ij}) >\max\left\{\beta_{ij}, \tau_i\right\}\right].
\end{equation}
    According to the definition of $\tau_i$, when $\tau_i > 0$,
    \begin{equation}\label{eq:tau3}
    \sum_j\Pr_{t_{ij}}\left[V_i(t_{ij}) >\max\left\{\beta_{ij}, \tau_i\right\}\right] = \half.        
    \end{equation}

    Combining~\eqref{eq:tau2}, \eqref{eq:tau3} and Lemma~\ref{lem:tailtau}, we then get $\sum_i\tau_i \le \frac{4}{1-b}\cdot \rprev$.
\end{proof}

{It is evident that Lemma~\ref{lem:coreh} is a direct consequence of the amalgamation of Lemma~\ref{lem:utih}, Lemma~\ref{lem:entryfee}, and Lemma~\ref{lem:sumtau}. Analogously, by combining Lemma~\ref{lem:coregap} and Lemma~\ref{lem:coreh}, we subsequently obtain Lemma~\ref{lem:approxcore}.}

Finally, we are now ready to prove our main theorem, {i.e.,} Theorem~\ref{thm:main}.
\begin{proof}[Proof of Theorem~\ref{thm:main}]
From the statement of Lemma~\ref{lem:approxsingle}, Lemma~\ref{lem:approxtail} and Lemma~\ref{lem:approxcore}, we get that 
\begin{align}\label{eq:mainproof1}
\begin{split}    
\single\InParentheses{\sigma^{(\beta)}, \beta} &\leq 8 \cdot \rprev\\
\tail\InParentheses{\beta} &\leq \frac{2}{1 - b} \cdot \rprev\\
\core(\sigma^{(\beta)},\beta) &\leq  \frac{4}{c}\cdot\efrev^{(\str)}_{\dist}\InParentheses{\auc} + \frac{1}{c}\cdot \rev^{(\str)}_{\dist}\InParentheses{\auc} + \InParentheses{\frac{2b + 2}{b(1-b)} + \frac{10}{c(1-b)}}\cdot \rprev
\end{split}
\end{align}

Lemma~\ref{lem:RevUpper} demonstrates that 
\begin{align}\label{eq:mainproof2}
\begin{split}
    \rev_{\dist}(\M) \leq 2\cdot \single\InParentheses{\sigma^{(\beta)}, \beta} + 4\cdot \tail\InParentheses{\beta} + 4\cdot \core\InParentheses{\sigma^{(\beta)}, \beta}.
\end{split}
\end{align}

Combining~\eqref{eq:mainproof1} and \eqref{eq:mainproof2}, we then get 
\[\rev_{\dist}(\M)\leq  \frac{16}{c}\cdot\efrev^{(\str)}_{\dist}\InParentheses{\auc} + \frac{4}{c}\cdot \rev^{(\str)}_{\dist}\InParentheses{\auc} + \left(\frac{16b + 8}{b(1-b)} + \frac{40}{c(1-b)} + { 16}\right)\cdot \rprev.\]

By \Cref{coro:rprev} and Lemma~\ref{lem:entryrev}, we then know that there exists a set of entry fees $\{e_i\}_{i\in [n]}$ and a set of reserve prices $\{r_{ij}\}_{i\in [n],j\in [m]}$ so that for any equilibrium $\str$ of auction $\auc$ with reserve price $r$, i.e., $\auc_{\mathrm{RP}}^{(r)}$, and any $\varepsilon_1,\varepsilon_2, \delta\in (0, 1)$, it holds that
\[\rev_{\dist}(\M)\leq \frac{20}{c\cdot (1 - \delta -\varepsilon_1)}\cdot \rev^{(\str)}_{\dist}\InParentheses{\auc^{(e)}_{\mathrm{EF}}} + \InParentheses{1 - \varepsilon_2}^{-1}\left(\frac{16b + 8}{b(1-b)} + \frac{40}{c(1-b)} + { 16}\right)\cdot \rev^{(\str')}_{\dist}\InParentheses{\auc_{\mathrm{RP}}^{(r)}}.\]

Taking \(\delta = \varepsilon_1 = \varepsilon_2 = 0.01\) and $b = \frac{1}{5}$, we get that 
\[\rev_{\dist}(\M)\leq \frac{21}{c}\cdot \rev^{(\str)}_{\dist}{ \InParentheses{\auc^{(e)}_{\mathrm{EF}}}} + \left(87 + \frac{51}{c}\right)\cdot \rev^{(\str')}_{\dist}\InParentheses{\auc_{\mathrm{RP}}^{(r)}}.\]

{Since this inequality holds for any BIC mechanism $M$, we have proved our claim.}
\end{proof}

\bibliographystyle{plain}
\bibliography{bibliography,references_yang, ref_tmp}

\appendix

\section{Additional Preliminaries}
\label{appendix:example}
\paragraph{Bayes-Nash Equilibrium} A strategy profile $\str = \InParentheses{\str_1, \str_2, \cdots, \str_n}$ is a Bayes-Nash equilibrium (BNE) with respect to type distribution $\dist$ and valuation functions $\{v_i\}_{i\in[n]}$ if and only if for any bidder $i$, any type $t_i$, and any {strategy $\tilde{\str}_i:T_i\rightarrow \mathbb{R}^m_{\geq 0}$}, the following inequality holds
\begin{large}
\begin{align*}
    \E_{\type_{-i} \sim \dist_{-i}}\InBrackets{\E_{\bid\sim \left(\str_i(\type_i), \str_{-i}(\type_{-i})\right)} \InBrackets{u_i\InParentheses{t_i, b}}} \geq  \E_{\type_{-i} \sim \dist_{-i}}\InBrackets{\E_{\substack{\tilde{\bid}_i \sim {\tilde{\str}_i(t_i)}\\ \bid_{-i}\sim s_{-i}\InParentheses{\type_{-i}} }}\InBrackets{u_i\InParentheses{t_i,\InParentheses{\tilde{\bid}_i, \bid_{-i}}}}}.
\end{align*}
\end{large}

\paragraph{Examples of~Valuations} Suppose $t = \InAngles{t_j}_{j\in [m]}$ where $\type_j$ is drawn independently from $\dist_j$. We {show} how subadditive functions over independent items {capture} various families of valuation functions.

\begin{itemize}
    \item Additive: $\type_j$ is the value of item $j$, and $\val\InParentheses{\type, S} = \sum_{j\in S} \type_j$.
    \item Unit-demand: $\type_j$ is the value of item $j$, and $\val\InParentheses{\type, S} = \max_{j\in S} \type_j$.
    \item Constrained Additive: $\type_j$ is the value of item $j$, and suppose $\mathcal{I}$ is {a} family of feasible sets. $\val(\type, S) = \max_{{Y\subseteq S, Y\in \mathcal{I}}} \sum_{i\in Y} \type_i$.
    \item XOS/Fractionally Subadditive: let $\type_j = \left\{\type_j^{(k)}\right\}_{k\in [K]}$ be the collection of values of item $j$ for each {of the $K$ additive functions}, and $\val(\type, S) = \max_{k\in [K]} \sum_{j\in S} \type_j^{(k)}$.
\end{itemize} 
\section{Tie-breaking and the Existence of Equilibrium}
\label{appendix:tiebreaking}

\subsection{Tie-breaking}

For distribution $D$ with point masses, the following reduction will convert it to {a} continuous one. We will overload the notation of $D$ and think of it as a bivariate distribution with the first coordinate drawn from the previous single-variate distribution $D$ and the second tie-breaker coordinate drawn independently and uniformly from $[0,1]$. And $(X_1, t_1) > (X_2, t_2)$ if and only if either $X_1 > X_2$, or $X_1 = X_2$ and $t_1 > t_2$. Since the tie-breaker coordinate is continuous, the probability of having $(X_1, t_1) = (X_2, t_2)$ for any two values during a run of any mechanism is zero. 

Remind the second coordinate is only used to break ties, and it does not affect the calculation of payment. Note that when we {run} a mechanism with entry fees $\{e_i\}$, the second coordinate does not affect whether bidder $i$ chooses to pay the entry fee or not. It is only used to {break ties} in the execution of $\auc$. This means that we can even remove the second coordinate when implementing the mechanism with entry fees and still use the same ways to {break ties as in $\auc$}. Therefore, by adding the second tie-breaker coordinate, we get a continuous distribution, and do not change the structure of equilibrium of mechanisms with entry fees.

\subsection{The Existence of Equilibrium}
\label{subsec:existence}

Our result applies to every equilibrium in simultaneous auctions that satisfies $c$-efficient. However, equilibria may not exist when the type spaces and and strategy spaces are both continuous. To fix this, we can restrict the strategy spaces to be discrete and bounded, {e.g.,} $\varepsilon$-grid in $[0,H]$, and assume the type spaces to be finite. Consequently, this transforms the {game into a finite one}, and thus an equilibrium must inherently exist.

We refer readers to \cite{feldman_simultaneous_2013} {for} a detailed discussion of existence of equilibrium in simultaneous auctions.

\section{Proof of Lemma~\ref{lem:firstprice} and \ref{lem:allpay}}
\label{appendix:proofofgood}

The proof here is inspired by \cite{feldman_simultaneous_2013}.

The first and second condition is obviously true for simultaneous first-price {auctions} and simultaneous all-pay {auctions}. Now we argue that the third condition with $c = \half$ is satisfied by simultaneous first-price {auctions} and simultaneous all-pay {auctions}. Consider any bidder $i$ with type $t_i$ and a set of items $S\subseteq [m]$.

We let $P_{-i}$ be the distribution of $m$-dimensional vector $\max_{i\ne i}\bid_{i'} = \left(\max_{i'\ne i}\bid_{i'}^{(j)}\right)_{j\in [m]}$ where the randomness is from both $t_{-i}$ and $\str_{-i}(t_{-i})$. {Let \( q_i \) be a random variable sampled from the distribution \( P_{-i} \). Consider the random bid of bidder \( i \), which is \( q_i \) plus a small constant \( \varepsilon > 0 \) added to each component, with the entire vector constrained to the set \( S \).}
 For ease of notation, we denote this vector by $(q_i+\varepsilon)|_S$, whose $j$-th coordinate is $q_i^{(j)}+\varepsilon$ when $j\in S$, and equals to $0$ otherwise.

\begin{align}
    \begin{split}
        & \E_{\substack{\type_{-i} \sim \dist_{-i}\\ q_i\sim P_{-i},~ \bid_{-i}\sim \str_{-i}\InParentheses{\type_{-i}}}}\InBrackets{\val_i\InParentheses{t_i,\alloc_i\InParentheses{(q_i+\varepsilon)|_S, \bid_{-i}\InParentheses{\type_{-i}}} \cap S}}\\
        & \ge \E_{\substack{\type_{-i} \sim \dist_{-i}\\ q_i\sim P_{-i},~ \bid_{-i}\sim \str_{-i}\InParentheses{\type_{-i}}}}\InBrackets{\val_i\InParentheses{t_i, \left\{j: q_i^{(j)} + \varepsilon> \max_{i'\ne i}b_{i'}^{(j)}\right\}\cap S} }\\
        & = \E_{\substack{q_i\sim P_{-i}\\r_i\sim P_{-i}}}\InBrackets{\val_i\InParentheses{t_i, \left\{j: q_i^{(j)} + \varepsilon > r_i^{(j)}\right\}\cap S} }\\
        & = \half \cdot \E_{\substack{q_i\sim P_{-i}\\r_i\sim P_{-i}}}\InBrackets{\val_i\InParentheses{t_i, \left\{j: q_i^{(j)} + \varepsilon > r_i^{(j)}\right\}\cap S}  + \val_i\InParentheses{t_i, \left\{j: r_i^{(j)} + \varepsilon > q_i^{(j)}\right\}\cap S} }\\
        & \ge \half \val_i\InParentheses{t_i, S}.
    \end{split}
\end{align}

The last inequality is because the union of $\left\{j: q_i^{(j)} + \varepsilon > r_i^{(j)}\right\}$ and $\left\{j: r_i^{(j)} + \varepsilon > q_i^{(j)}\right\}$ is $[m]${, and $v_i(t_i,\cdot)$ is a subadditive function}. Also notice that in simultaneous first-price or all-pay auction{s}, the payment on a single item {does not exceed} the bid on {the} item, so the total payment of a bidder {does} not exceed the sum of their bid{s}.
\begin{align}
    \begin{split}
        \uti_i\InParentheses{\type_i, S} & \ge \E_{\substack{\type_{-i} \sim \dist_{-i}\\ q_i\sim P_{-i},~ \bid_{-i}\sim \str_{-i}\InParentheses{\type_{-i}}}}\InBrackets{{\val_i\InParentheses{t_i,\alloc_i\InParentheses{q_i|_S, \bid_{-i}} \cap S}} - {\sum_{j\in S}\pay_i^{(j)}\InParentheses{q_i^{(j)}, \bid_{-i}^{(j)}}}}\\
        & \ge \half \val_i\InParentheses{\type_i, S} - \sum_{j\in S}\E_{q_i \sim P_{-i}}{\InBrackets{q_i^{(j)}}} - \left|S\right| \cdot \varepsilon\\
        & = \half \val_i\InParentheses{\type_i, S} - \sum_{j\in S}\E_{\substack{\type_{-i} \sim \dist_{-i}\\ \bid_{-i}\sim \str_{-i}\InParentheses{\type_{-i}}}}\InBrackets{\max_{i'\ne i}{\bid_{i'}^{(j)}}} - \left|S\right| \cdot \varepsilon.\\
    \end{split}
\end{align}

At the end, since in first-price or all-pay auction the revenue from a item is at least the maximum of bid on this item, so
\[\rev^{(\bid)}\InParentheses{\auc, \left\{j\right\}} \ge \E_{\substack{\type \sim \dist\\ \bid \sim \str\InParentheses{\type}}}\InBrackets{{\max_i\bid_i^{(j)}}}.\]

Therefore,
\begin{align*}
    \uti_i\InParentheses{\type_i, S} & \ge \half \val_i\InParentheses{\type_i, S} - \sum_{j\in S}\E_{\substack{\type_{-i} \sim \dist_{-i}\\ \bid_{-i}\sim \str_{-i}\InParentheses{\type_{-i}}}}\InBrackets{\max_{i'\ne i}{\bid_{i'}^{(j)}}} - \left|S\right| \cdot \varepsilon\\
    & \ge \half \val_i\InParentheses{\type_i, S} - \sum_{j\in S}\E_{\substack{\type \sim \dist\\ \bid \sim \str\InParentheses{\type}}}\InBrackets{\max_{i}{\bid_i^{(j)}}} - \left|S\right| \cdot \varepsilon\\
    & \ge \half \val_i\InParentheses{\type_i, S} - \rev_{\dist}^{(\bid)}\InParentheses{\auc, S} - \left|S\right| \cdot \varepsilon.
\end{align*}

Taking $\varepsilon \rightarrow 0$, by definition of $\uti_i\InParentheses{\type_i, S}$, we know
\[\uti_i\InParentheses{\type_i, S} \ge \half \val_i\InParentheses{\type_i, S} - \rev_{\dist}^{(\bid)}\InParentheses{\auc, S}.\] 
\section{Missing Proofs in Section~\ref{sec:mecha}}
\label{appendix:mecha}

\subsection{Proof of Lemma~\ref{lem:RandomSameEq}}
\label{subsec:proofofRandomSameEq}
For any strategy profile $s$ with respect to a prior distribution of types $\dist$ in auction {$\auc$}, {we slightly abuse notation and }let $u_i^{(\str)}(t_i)$ be the interim utility of bidder $i$ with {type} $t_i$. Namely, 
\[u_i^{(\str)}\InParentheses{t_i} = \E_{\type_{-i} \sim \dist_{-i}}\InBrackets{\E_{\bid\sim \str(\type_i, \type_{-i})}\InBrackets{u_i\InParentheses{t_i, b}}}.\]

Then by definition a strategy profile $\str$ is a {Bayes-Nash} equilibrium in {$\auc$} iff for any bidder $i$, type $\type_i$ and a mixed strategy $\str_i'$, $u_i^{(\str)}\InParentheses{t_i}\ge u_i^{(\str_i', \str_{-i})}\InParentheses{t_i}$.

Given a strategy profile $\str$ in auction {$\auc_{\mathrm{EF}}^{(e)}$}, for the bidder $i$ with type $\type_i$, {$i$ receives $\delta$ times their interim utility $u_i^{(\str)}(t_i)$ in auction $\auc$ by reporting $z_i = 0$} If {$i$ reports} $z_i = 1$, the interim utility is $u_i^{(\str)}(t_i)$ minus {$(1-\delta)e_i$}. Hence, in auction {$\auc_{\mathrm{EF}}^{(e)}$} the interim utility of bidder $i$ with type $\type_i$ {is} 
\[\tilde{u}_i^{(\str)}\InParentheses{t_i} := \max\left\{
\delta\cdot u_i^{(\str)}(t_i), u_i^{(\str)}(t_i) - (1 - \delta) e_i\right\}\]

Notice that $\max\left\{\delta\cdot x, x - (1 - \delta) e_i\right\}$ is a strictly increasing function with respect to $x$ for $\delta\in (0, 1)$, which means that $\tilde{u}_i^{(\str)}\InParentheses{t_i}$ is a strictly increasing function with respect to $u_i^{(\str)}\InParentheses{t_i}$. Thus, $\tilde{u}_i^{(\str)}\InParentheses{t_i}\ge \tilde{u}_i^{(\str_i', \str_{-i})}\InParentheses{t_i}$ is equivalent to $u_i^{(\str)}\InParentheses{t_i}\ge u_i^{(\str_i', \str_{-i})}\InParentheses{t_i}$. As a result, we know that a strategy profile $\str$ is a equilibrium in {$\auc$} if and only if it is a equilibrium in {$\auc_{\mathrm{EF}}^{(e)}$}.

\subsection{Proof of Lemma~\ref{lem:entryrev}}
\label{subsec:proofofentryrev}
We use the same notation $u_i^{(\str)}(\type_i)$ to denote the interim utility of bidder $i$ with type $\type_i$ in auction $\auc${, when all bidders bid according to strategy profile $s$}.

Taking $e_i=0$ for all $i\in [n]$, we know $\rev^{(\str)}_{\dist}\InParentheses{\auc^{(e)}_{\mathrm{EF}}} = \rev^{(\str)}_{\dist}(\auc)$,

If $\efrev^{(\str)}_{\dist}(\auc) = 0$, we have already finished the proof.

When $\efrev^{(\str)}_{\dist}(\auc) > 0$, we only need to prove for any $\varepsilon > 0$, there exists a set of entry fees $\{e_i\}_{i\in [n]}$ so that
\[\rev^{(\str)}_{\dist}\InParentheses{\auc^{(e)}_{\mathrm{EF}}} \geq \left(1 - \delta - \varepsilon\right)\efrev^{(\str)}_{\dist}(\auc). \]

Now consider any $\varepsilon > 0$, by definition of $\efrev^{(\str)}_{\dist}(\auc)$, there exists a set of $e_i$ such that 
\[\sum_{i} e_i\cdot \Pr_{\substack{\type_i\sim \dist_i}}\InBrackets{u_i^{(\str)}(t_i) \geq e_i} \ge \left(1 - \varepsilon\right)\efrev^{(\str)}_{\dist}(\auc)\]

Now simply consider the mechanism $\auc$ with entry fee $\{e_i\}_{i\in [n]}$, {i.e.,} $\auc_{\mathrm{EF}}^{(e)}$. It's clear that bidder $i$ will pay entry fee iff  $u_i^{(\str)}(t_i) \ge e_i$. The revenue of $\auc^{(e)}_{\mathrm{EF}}$ is at least its revenue from entry fees, so
\begin{align*}
    \rev_{\dist}^{(\str)}\InParentheses{\auc_{\mathrm{EF}}^{(e)}} & \ge (1-\delta)\sum_{i} e_i\cdot \Pr_{\substack{\type_i\sim \dist_i}}\InBrackets{u_i^{(\str)}(t_i) \ge e_i}\\
    & \ge (1 - \delta - \varepsilon)\efrev^{(\str)}_{\dist}(\auc).
\end{align*}

By choosing the better entry fee between $0$ and $\{e_i\}_{i\in [n]}$, we conclude our proof.

\subsection{A Hard Instance for {the} Simultaneous Second Price Auction}
\label{subsec:hardinstance}

We first provide a counter-example to show that not every equilibrium of {the} simultaneous second price auction satisfies the third condition in \Cref{def:good}.
\begin{example}
\label{ex:S2A}
Consider the following deterministic instance. There are $n$ unit-demand bidders and $n$ items. For each bidder $i$, their favourite item is the $i$-th item, and their {value} towards that item is $1$. For any other item, their {value} is $\varepsilon$, where $\varepsilon$ is a {constant strictly less than $1$}. 
    
\begin{figure}[H]
\centering

\begin{tikzpicture}[thick,amat/.style={matrix of nodes,nodes in empty cells,
  fsnode/.style={draw,solid,circle,execute at begin node={}},
  ssnode/.style={draw,solid,circle,execute at begin node={}}}]

 \matrix (m1left) [amat,nodes=fsnode,label=above:Bidders,row sep=2.5em,dashed,draw = none,rounded corners]  {
  \node (m1-left1) [draw,fill=none] {$1$};\\
  \node (m1-left2) [draw,fill=none] {$2$};\\ 
  \node (m1-left3) [draw = none,fill=none] {$\cdots$};\\
  \node (m1-left4) [draw,fill=none] {${n}$};\\
 };

 \matrix (m1right) [amat,right=3cm of m1left,nodes=ssnode, draw = none, label=above:Items,row sep=2.5em,dashed,rounded corners]  {
 \node (m1-right1) [draw,fill=none] {$1$};\\
 \node (m1-right2) [draw,fill=none] {$2$};\\  
 \node (m1-right3) [draw = none,fill=none] {$\cdots$};\\
 \node (m1-right4) [draw,fill=none] {$n$};\\
 };

\draw (m1-left1) -- node[below] {} node[above] {$1$} ++(m1-right1);
\draw [dashed] (m1-left1) -- node[below] {} node[above] {$\varepsilon$} ++(m1-right2);
\draw [dashed](m1-left1) -- node[below] {} node[above] {$\varepsilon$} ++(m1-right4);

\draw [dashed](m1-left2) -- node[below] {} node[above] {$\varepsilon$} ++(m1-right1);
\draw  (m1-left2) -- node[below] {} node[above] {$1$} ++(m1-right2);
\draw [dashed](m1-left2) -- node[below] {} node[above] {$\varepsilon$} ++(m1-right4);

\draw [dashed] (m1-left4) -- node[below] {} node[above] {$\varepsilon$} ++(m1-right1);
\draw [dashed] (m1-left4) -- node[below] {} node[above] {$\varepsilon$} ++(m1-right2);
\draw (m1-left4) -- node[below] {} node[above] {$1$} ++(m1-right4);

\end{tikzpicture}

\caption{A Counter-Example for Simultaneous Second Price Auction}
\label{fig:hardinstance}
\end{figure}
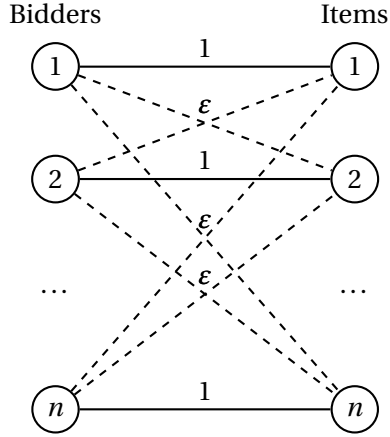
\end{example}

For this instance, suppose each {bidder} $i$ bids $1$ on their favorite item, {i.e.,} item $i$, and bids $0$ on any item else. It is clear  that this is a no over-bidding pure Nash equilibrium as {everyone} gets their favorite item and pays nothing. Therefore, in this equilibrium $\str$, $\rev_{\dist}^{(\str)}(\auc) = 0$. What's more, it is easy to see that this equilibrium is optimal in welfare. 

Let $S_i = [n]\backslash \{i\}$. However, we can see that $\mu_i^{(\str)}\InParentheses{\type_i,S_i} = 0$ as for every item $j \neq i$, the maximum bid at equilibrium $\str$ is $1$, and consequently, bidder $i$ has no motivation to engage in competition for that item. Also note that $\val_i(\type_i, S_i) = \varepsilon$. This implies that \[\mu_i^{(\str)}\InParentheses{\type_i,S_i} + \rev_{\dist}^{(\str)}(\auc)  = 0  < c\cdot \val_i(\type_i, S_i) \]
for any $c > 0$.

\subsection{A More Detailed Discussion of \SAprop~simultaneous auctions}
\label{subsec:gooddiscussion}

{We} introduce a property of $\mu^{(\str)}(\type_i,\cdot)$ that is essential in approximating the optimal revenue.

\begin{lemma}
\label{lem:subadditiveofmu}

For any $i$ and any constant $l_i \geq 0$, let $L_i(t_i)$ be the set {$\{j: V_i(\type_{ij}) < l_i\}$}, and define $\hat{\uti}_i^{(\str)}\InParentheses{t_i,S} = \uti_i^{(\str)}\InParentheses{t_i, S\cap L_i(t_i)}$. {Recall that $\uti_i^{(\str)}\InParentheses{t_i, S\cap L_i(t_i)}$ is defined in \Cref{def:good}.}
If the first condition of {\Cref{def:good}} is satisfied by \CEtuple,  $\hat{\uti}^{(\str)}_i\InParentheses{\cdot, \cdot}$ satisfies monotonicity, subadditivity, no externalities and {is} $l_i$-Lipschitz.
\end{lemma}

\begin{proof}
We first prove $\uti_i$ satisfies monotonicity, subadditivity and no externalities.

For any types $t_i,t_i'$, such that $t_{ij} = t_{ij}'$ for all $j\in S$,
\begin{align*}
    \uti_i^{(\bid)}\InParentheses{\type_i, S} & = \sup_{q_i}\ \E_{\substack{\type_{-i} \sim \dist_{-i}\\ \bid_{-i}\sim \str_{-i}\InParentheses{\type_{-i}}}}\InBrackets{{\val_i\InParentheses{t_i,\alloc_i\InParentheses{q_i, \bid_{-i}} \cap S} - { \sum_{j\in S} \pay_i^{\left(j\right)}\InParentheses{q_i^{\left(j\right)}, \bid_{-i}^{\left(j\right)}}}}}\\
    & = \sup_{q_i}\ \E_{\substack{\type_{-i} \sim \dist_{-i}\\ \bid_{-i}\sim \str_{-i}\InParentheses{\type_{-i}}}}\InBrackets{{\val_i\InParentheses{t_i',\alloc_i\InParentheses{q_i, \bid_{-i}} \cap S} -  { \sum_{j\in S} \pay_i^{\left(j\right)}\InParentheses{q_i^{\left(j\right)}, \bid_{-i}^{\left(j\right)}} }}}\\
    & = \uti_i^{(\bid)}\InParentheses{\type_i', S}.
\end{align*} where second equality is by no externalities of $\val_i$. Thus, $\uti_i$ has no externalities.

For any set $U\subseteq V \subseteq [m]$,
\begin{align*}
    \uti_i^{(\bid)}\InParentheses{\type_i, U } & = \sup_{q_i}\ \E_{\substack{\type_{-i} \sim \dist_{-i}\\ \bid_{-i}\sim \str_{-i}\InParentheses{\type_{-i}}}}\InBrackets{{\val_i\InParentheses{t_i,\alloc_i\InParentheses{q_i, \bid_{-i}} \cap U} -  {\sum_{j\in S} \pay_i^{\left(j\right)}\InParentheses{q_i^{\left(j\right)}, \bid_{-i}^{\left(j\right)}}}}}\\
    & \le \sup_{q_i}\ \E_{\substack{\type_{-i} \sim \dist_{-i}\\ \bid_{-i}\sim \str_{-i}\InParentheses{\type_{-i}}}}\InBrackets{{\val_i\InParentheses{t_i,\alloc_i\InParentheses{q_i, \bid_{-i}} \cap V} - {\sum_{j\in S} \pay_i^{\left(j\right)}\InParentheses{q_i^{\left(j\right)}, \bid_{-i}^{\left(j\right)}}}}}\\
    & = \uti_i^{(\bid)}\InParentheses{\type_i, V}.
\end{align*}
{The inequality} is because $\val_i$ is monotone. So $\uti_i$ is monotone.

We use {$q_i|_S$} to denote the bid vector $q_i$ {restricted to bundle} {$S$, which means that $(q_i|_S)^j$ equals $q_i^j$ when $j\in S$}, and equals to {the null action $\perp$} otherwise. For all $\type_i$ and $U, V\subseteq [m]$, let $W = \InParentheses{U\cup V}\backslash U$. Then  {$U\cap W = \emptyset$}, $U\cup W = U\cup V$ and  $W\subseteq V$. To prove subadditivity of $\uti_i$, we first {prove the following claims}, one equality and one inequality, which are true for any bid profile $b_{-i}$.
\begin{align}
    \begin{split}
    \label{eq:cupofalloc}
    \alloc_i\InParentheses{q_i, \bid_{-i}} \cap \InParentheses{U\cup W} & = \bigcup_{j\in U\cup W}\alloc_i^{(j)}\InParentheses{q_i^j, \bid_{-i}^{(j)}}\\
    & = \InParentheses{\bigcup_{j\in U}\alloc_i^{(j)}\InParentheses{q_i^j, \bid_{-i}^{(j)}}} \bigcup \InParentheses{\bigcup_{j\in W}\alloc_i^{(j)}\InParentheses{q_i^j, \bid_{-i}^{(j)}}}\\
    & = \InParentheses{\alloc_i\InParentheses{q_i|_U, \bid_{-i}} \cap U} \cup \InParentheses{\alloc_i\InParentheses{q_i|_W, \bid_{-i}} \cap W}
    \end{split}
\end{align}

Therefore,
\begin{align*}
    \uti_i\InParentheses{\type_i, U\cup V} = & \uti_i\InParentheses{\type_i, U\cup W}\\
    = & \sup_{q_i}\ \E_{\substack{\type_{-i} \sim \dist_{-i}\\ \bid_{-i}\sim \str_{-i}\InParentheses{\type_{-i}}}}\InBrackets{{\val_i\InParentheses{t_i,\alloc_i\InParentheses{q_i, \bid_{-i}} \cap \InParentheses{U\cup W}} -  {\sum_{j\in U \cup W}\pay_i^{(j)}\InParentheses{q_i^{(j)}, \bid_{-i}^{(j)}}}}}\\
    \le & \sup_{q_i}\ \vast\{\E_{\substack{\type_{-i} \sim \dist_{-i}\\ \bid_{-i}\sim \str_{-i}\InParentheses{\type_{-i}}}}\InBrackets{{\val_i\InParentheses{t_i,\alloc_i\InParentheses{q_i|_U, \bid_{-i}} \cap U} -  {\sum_{j\in U}\pay_i^{(j)}\InParentheses{q_i^{(j)}, \bid_{-i}^{(j)}}}}}\\
    & + \E_{\substack{\type_{-i} \sim \dist_{-i}\\ \bid_{-i}\sim \str_{-i}\InParentheses{\type_{-i}}}}\InBrackets{{\val_i\InParentheses{t_i,\alloc_i\InParentheses{q_i|_W, \bid_{-i}} \cap W} -  {\sum_{j\in W}\pay_i^{(j)}\InParentheses{q_i^{(j)}, \bid_{-i}^{(j)}}}}}\vast\}\\
    \le & \sup_{q_i}\ \E_{\substack{\type_{-i} \sim \dist_{-i}\\ \bid_{-i}\sim \str_{-i}\InParentheses{\type_{-i}}}}\InBrackets{{\val_i\InParentheses{t_i,\alloc_i\InParentheses{q_i, \bid_{-i}} \cap U} - {\sum_{j\in U} \pay_i^{(j)}\InParentheses{q_i^{(j)}, \bid_{-i}^{(j)}}}}}\\
    & + \sup_{q_i} \E_{\substack{\type_{-i} \sim \dist_{-i}\\ \bid_{-i}\sim \str_{-i}\InParentheses{\type_{-i}}}}\InBrackets{{\val_i\InParentheses{t_i,\alloc_i\InParentheses{q_i, \bid_{-i}} \cap W} -  {\sum_{j\in W} \pay_i^{(j)}\InParentheses{q_i^{(j)}, \bid_{-i}^{(j)}}}}}\\
    = & \uti_i\InParentheses{\type_i, U} + \uti_i\InParentheses{\type_i, W}\\
    \le & \uti_i\InParentheses{\type_i, U} + \uti_i\InParentheses{\type_i, V}.
\end{align*}
{The} first inequality is by \eqref{eq:cupofalloc}, {the subadditivity of $\val_i$, and the fact that \(U\cap W = \varnothing\)}. The second inequality is from the property of {the} $\sup$ operator, and the third inequality is because $\uti_i$ is monotone. Thus, $\uti_i$ is subadditive.

Consider any constant $l_i$ {in the} definition of $\utih_i$. {By the monotonicity and subadditivity of $\uti_i$}, we can directly {conclude that $\utih_i$ is also monotone and subadditive}.

For any types $t_i,t_i'$, such that $t_{ij} = t_{ij}'$ for all $j\in S$, we know $S\cap L_i(\type_i) = S\cap L_i(\type_i')$, since for any $j\in S$, 
\[j\in L_i(\type_i) \Leftrightarrow V_i(\type_{ij}) < l_i \Leftrightarrow V_i(\type_{ij}') < l_i \Leftrightarrow j\in L_i(\type_i')\]

Hence {\[\utih_i(t_i, S) = \uti_i(t_i, S\cap L_i(t_i)) = \uti_i(t_i', S\cap L_i(t_i')) = \utih_i(t_i', S).\]}

Thus, {$\utih_i$} satisfies monotonicity, subadditivity and no externalities.

Finally, we prove {$\utih_i$} is $l_i$-Lipschitz.

For any $\type_i, \type_i'\in T_i$, and set $X,Y\subseteq[m]$, define set $H = \left\{j\in X\cap Y ~:~ \type_{ij} = \type_{ij}'\right\}$. Because of the {no externalities property} of $\utih_i$, we know $\utih_i(\type_i, H) = \utih_i(\type_i', H)$.

\begin{align*}
    \left|\utih_i\InParentheses{\type_i, X} - \utih_i\InParentheses{\type_i', Y} \right| & = {\max}\left\{\utih_i\InParentheses{\type_i, X} - \utih_i\InParentheses{\type_i', Y}, \utih_i\InParentheses{\type_i', Y} - \utih_i\InParentheses{\type_i, X} \right\}\\
    & \le {\max}\left\{\utih_i\InParentheses{\type_i, X} - \utih_i\InParentheses{\type_i', H}, \utih_i\InParentheses{\type_i', X} - \utih_i\InParentheses{\type_i, H}\right\}\\
    & \le {\max}\left\{\utih_i\InParentheses{\type_i, X\backslash H}, \utih_i\InParentheses{\type_i', Y\backslash H}\right\}\\
    & = {\max}\left\{\uti_i\InParentheses{\type_i, \InParentheses{X\backslash H}\cap L_i\InParentheses{\type_i}}, \uti_i\InParentheses{\type_i', \InParentheses{Y\backslash H}\cap L_i\InParentheses{\type_i'}}\right\}\\
    & \le l_i \cdot \max\left\{|X\backslash H|, |Y\backslash H|\right\}\\
    & \le l_i \cdot \InParentheses{\left|X\Delta Y\right| + \left|X\cap Y\right| - \left|H\right|}.
\end{align*}

\end{proof}

In the following, we show that for any \CEtuple~that satisfies the third condition, it also achieves a high welfare at the equilibrium $\str$. Let us define {$\mathrm{Wel}^{(\str)}_{\dist}(\auc)$ as the social welfare of auction $\auc$ at $\str$}, and $\OPT_i(\type)$ as the set of items allocated to bidder $i$ in the allocation that maximizes social welfare when the bidders' types are $\type$. We give a formal proof that the welfare at $\str$ is at least $c$ fraction of the optimal welfare: 
\begin{align*}
    \mathrm{Wel}^{(\str)}_{\dist}(\auc) &= \sum_{i\in [n]} \E_{\substack{\type\sim \dist\\ \bid\sim \str(\type)}}\InBrackets{u_i\InParentheses{\type_i, b}} + \rev_{\dist}^{(\str)}(\auc)\\
    & = \E_{\type\sim \dist} \InBrackets{\sum_{i\in [n]} \mu_i^{(\str)}(\type_i, [m]) + \rev_{\dist}^{(\str)}\left(\auc, [m]\right)}\\
    & \geq \E_{\type\sim \dist} \InBrackets{\sum_{i\in [n]} \left(\mu_i^{(\str)}(\type_i, \OPT_i(\type)) + \rev_{\dist}^{(\str)}(\auc, \OPT_i(\type)\right)} \\
    & \geq c\cdot \E_{\type\sim \dist} \InBrackets{\sum_{i\in [n]} \val_i(\type_i, \OPT_i(t))}
\end{align*}
{The second equation holds since $\str$ is a Bayes-Nash equilibrium.}
The first inequality comes from the monotonicity of $\mu_i^{(\str)}(\type_i,\cdot)$ which is proved in Lemma~\ref{lem:subadditiveofmu} and the second inequality directly follows from the third condition. 

\notshow{
Finally, we prove the that in the single-item setting, the notation proposed by Hartline et al.~\cite{jason} implies the thrid condition. For any single item mechanism $A$ and its equilibrium $s$ that is $(\eta, \mu)$-individual and competitive efficient, by definition, for any bidder $i$ with valuation $v_i$,

\[u_i^{(\str)}(\val_i) + \E\InBrackets{\hat{b}_i} \ge \eta\val_i,\]
where $\hat{b}_i$ is random threshold bid for bidder $i$, and 
\[\rev^{(\str)}(A) \ge \mu \cdot \E_{\substack{\val_{-i} \sim \dist_{-i}\\ \bid_{-i}\sim \str_{-i}\InParentheses{\val_{-i}}}}\InBrackets{\max_{i\in [n]}\left\{\hat{b}_i\right\}}.\]

Combining them directly we know $u_i^{(\str)}(\val_i) + \frac{1}{\mu}\rev^{(\str)}(A) \ge \eta\val_i$, and by $\mu\le 1$, we get $u_i^{(\str)}(\val_i) + \rev^{(\str)}(A) \ge \eta\mu\val_i$, which means that the third condition in \SAprop~ with $c=\eta\mu$ holds.
}

\subsection{Proof of Lemma~\ref{lem:reserveprice}}
\label{subsec:proofofreserveprice}
\begin{prevproof}{Lemma}{lem:reserveprice}
    Notice that by the first condition and {the} union bound, for any item $j$, the probability  that each bidder $i$'s {value} on item $j$ is smaller {than their} reserve price on item $j$, $r_{i,j}$, is at least $1-\sum_{i\in [n]}\Pr[V_{i}(t_{ij}) \ge r_{ij}]\ge 1-b$. Similarly, by the second condition, we know that for any bidder $i$, the probability that {their value} of any item $j$ is below the reserve price $r_{ij}$ is at least $\half$.

    We first prove that for any equilibrium $s$ of $\auc^{(r)}_{\mathrm{RP}}$, any bidder $i$ will always take {the null action $\perp$} when {their value} on this item is smaller than the reserved price. Suppose there exists a bid equilibrium that does not follow this. For any $j\in[m]$ let {$I_j=\{i:\Pr_{t_{i}, b_i\sim s_i(t_{i})}[V_i(t_{ij}) < r_{ij} \wedge b_i^j \ne \perp] > 0\}$} be the set of {bidders that have a non-zero probability to compete for item $j$  while their value is less than the reserve.} Assume that $I_k$ is non-empty {for some $k$}. {Consider the event that satisfies the following: (i) for any bidder $i\notin I_k$, $i$'s value on item $k$ is strictly less than $r_{ik}$; (ii) for any bidder $i\in I_k$, $V_i(t_{ik}) < r_{ik}$ and $b_i^k\neq \perp$. It is not hard to see that this event happens with non-zero probability. Conditioning on this event, the winner of item $k$ must be some bidder $i^*$ in $I_k$. We argue that $i^*$'s expected utility is strictly worse  compared to the scenario where their bids remain unchanged for other items, and $b_{i^*}^k$ is replaced with $\perp$. The reason is that $i^*$ has a subadditive valuation, so $i^*$'s  utility is strictly worse after acquiring item $k$ at a price larger than $V_{i^*}(t_{i^*k})$.}

    Now consider bidder $i$ with type $t_i$ satisfying two conditions (i) $V_i(t_{ij})\ge r_{ij}$, (ii) $\forall k \ne j, V_i(t_{ij}) < r_{ik}$. Then {, as we argued in the previous paragraph, $i$ will take the null action $\perp$} on items other than $j$. Now since bidding {$r_{ij}$} for item $j$ {will give $i$ a non-negative utility, $i$ will not bid $\perp$ for item $j$.} Further consider (iii) $\forall i'\ne i, V_{i'}(t_{i'j}) < r_{i'j}$ which {implies that any bidder other than $i$ bill bid $\perp$ for item $j$}. Then if all of (i), (ii), (iii) {holds}, bidder $i$ will {receive} item $j$. The probability of (ii) and (iii) {holds} is greater than $\half$ and $1-b$ by the first paragraph. Because conditions (i), (ii) and (iii) are independent, bidder $i$ wins item $j$ with probability at least $\frac{1-b}{2}\cdot \Pr[V_i(t_{ij})\ge r_{ij}]$. Thus the expected revenue of {the mechanism} is at least $\frac{1-b}{2}\cdot \sum_{i,j}r_{ij}\cdot \Pr[V_i(t_{ij})\ge r_{ij}]$.
    
\end{prevproof} 
\section{Missing Proofs in Section~\ref{sec:duality}}
\label{sec:dualityappendix}

\notshow{
\subsection{Proof of Lemma~\ref{lem:goodbeta}}
\label{subsec:goodbetaappendix}

We define
\[\beta_{ij} := \inf\left\{x\ge 0: \Pr_{t_{ij}}\InBrackets{V_i(t_{ij})\ge x}\le b\cdot \sum_{t_i\in T_i}f_i(t_i)\cdot \pi_{ij}(t_i)\right\},\] where $\pi_{i}(t_{ij}) = \sum_{S: j\in S}\sigma_{iS}(t_i)$. Since the distribution of $V_i(t_{ij})$ is continuous, $\Pr_{t_{ij}}\InBrackets{V_i(t_{ij})\ge \beta_{ij}}$ is exactly $b\cdot \sum_{t_i\in T_i}f_i(t_i)\cdot \pi_{ij}(t_i)$, and therefore for any $j\in [m]$,
\[\sum_i\Pr_{t_{ij}}\InBrackets{V_i(t_{ij})\ge \beta_{ij}} = b\cdot \sum_i\sum_{t_i\in T_i}f_i(t_i)\cdot \pi_{ij}(t_i) \le b.\]
So the first condition holds. The second condition is satisfied because by the first property in Lemma $2$ of \cite{cai_simple_2019}, $\sum_i\sum_{t_i\in T_i}f_i(t_i)\cdot \pi_{ij}^{(\beta)}(t_i) \le \sum_i\sum_{t_i\in T_i}f_i(t_i)\cdot \pi_{ij}(t_i)$.

}

\subsection{Proof of Lemma~\ref{lem:approxsingle}}
\label{subsec:singleappendix}

\begin{prevproof}{Lemma}{lem:approxsingle}
Our proof here is very similar to the proof of Lemma~13 in \cite{cai_simple_2017}. We introduce the single-dimensional copies setting defined in \cite{chawla_multi-parameter_2010}. In this setting, there are $nm$ agents, in which each agent $(i,j)$ has a value of $V_i(\type_{ij})$ of being served with $\type_{ij}$ sampled from $\dist_{ij}$ independently. The allocation must be a matching, meaning that for each $i\in [n]$, there is at most one $k\in [m]$ so that $(i,k)$ is served, and for each $j\in [m]$, there is at most one $k\in [n]$ so that $(k, j)$ is served. Fix the distribution $\dist$ and valuation function $V_i(\cdot)$, we denote the optimal BIC revenue in this setting as $\mathrm{OPT}^{\textsc{Copies-UD}}$. In \cite{cai_simple_2017}, they prove that for any $\sigma, \beta$, $\single\InParentheses{\sigma,\beta} \leq \mathrm{OPT}^{\textsc{Copies-UD}}$.  

For every $i\in[n], j\in [m]$, let $q_{ij}$ be the ex-ante probability that $(i,j)$ is served in the Myerson’s auction for the above copies settings. By definition, we have $\sum_j q_{ij}\le 1, \forall i\in[n]$ and $\sum_i q_{ij}\le 1, \forall j\in[m]$. 

The ironed virtual welfare contributed from $(i,j)$ is at most $\tilde{R}_{ij}(q_{ij})$, where $\tilde{R}_{ij}$ is the ironed revenue curve of $R_{ij}(q) = q\cdot F_{ij}^{-1}(1-q)${, where $F_{ij}$ is the CDF for the random variable $V_i(t_{ij})$, and $F_{ij}^{-1}$ is the corresponding quantile function}. Thus, there exist two quantiles $q_{ij}'$ and $q_{ij}''$, and a pair of corresponding convex representation coefficients $x_{ij}+y_{ij}=1$, such that $\tilde{R}_{ij}(q_{ij}) = x_{ij}\cdot R_{ij}(q_{ij}') + y_{ij}\cdot R_{ij}(q_{ij}'')$ and $q_{ij} = x_{ij}\cdot q_{ij}' + y_{ij}\cdot q_{ij}''$. Hence,
\begin{align}
    \begin{split}
    \label{eq:optud}
    \mathrm{OPT}^{\textsc{Copies-UD}} & \le \sum_{i,j}\tilde{R}_{ij}(q_{ij})\\
    & = \sum_{i,j}x_{ij}\cdot R_{ij}(q_{ij}') + y_{ij}\cdot R_{ij}(q_{ij}'')\\
    & \le 2\cdot \sum_{i, j} \InParentheses{x_{ij}\cdot \frac{q_{ij}'}{2}\cdot F_{ij}^{-1}\left(1-q_{ij}'/2\right) + y_{ij}\cdot \frac{q_{ij}''}{2}\cdot F_{ij}^{-1}\left(1-q_{ij}''/2\right)}\\
    & = 2\cdot \sum_{i, j} \E_{p_{ij}}\InBrackets{p_{ij}\cdot \Pr\InBrackets{V_{i}(t_{ij})\ge p_{ij}}}.
    \end{split}
\end{align}
$p_{ij}$ is a random price which equals to $F_{ij}^{-1}\left(1-q_{ij}'/2\right)$ with probability $x_{ij}$ and equals to $F_{ij}^{-1}\left(1-q_{ij}''/2\right)$ with probability $y_{ij}$. The second inequality here is because $F^{-1}(1-q) \le F^{-1}(1-q/2)$ for any CDF function $F$. To upper bound $\sum_{ij} \E_{p_{ij}}\InBrackets{p_{ij}\cdot \Pr\InBrackets{V_{i}(t_{ij})\ge p_{ij}}}$, we introduce an extension of lemma~ \ref{lem:reserveprice}.
\begin{lemma}
\label{lem:randomreserveprice}
    For a type distribution $\dist$, suppose simultaneous auction $\auc$ satisfies the first and second condition of \Cref{def:good}, and $\{r_{ij}\}_{i\in[n],j\in[m]}$ is a set of \textbf{independent random} prices that satisfy the following for some constant $b\in(0,1)$,
    \begin{itemize}
        \item[(1)] \(\sum_{i\in [n]}\Pr_{r_{ij}, t_{ij}}[V_{i}(t_{ij})\ge r_{ij}]\le b, \ \forall j\in [m] \);
        \item[(2)] \(\sum_{j\in[m]}\Pr_{r_{ij}, t_{ij}}[V_{i}(t_{ij})\ge r_{ij}]\le \half, \ \forall i\in [n] \).
    \end{itemize}

    Then for any Bayes-Nash equilibrium strategy profile $\str$ of simultaneous auction $\auc$ with independently randomized reserve price $r_{ij}$,
    \[\sum_{i,j}\E_{r_{ij}}\InBrackets{r_{ij}\cdot \Pr\left[V_i(t_{ij})\ge r_{ij}\right]} \le \frac{2}{1-b}\cdot \rev^{(s)}_{D}\left(\auc^{(r)}_{\mathrm{RP}}\right) \]
\end{lemma}
    To be more clear, the simultaneous mechanism with randomized reserve price, $\auc^{(r)}_{\mathrm{RP}}$, is defined to be the mechanism that first publicly independently draws $r_{ij}$ for each $i\in [n]$ and $j\in [m]$, and then implements the simultaneous auction $\auc$ with realized reserve prices $r_{ij}$'s. This is a distribution of simultaneous auctions with deterministic reserved prices, and thus its expected revenue is the expectation of these deterministic reserve prices auctions and is not larger than $\rprev$.

\begin{proof}
    For the similar reason in Lemma~ \ref{lem:reserveprice}, by the first condition, for any item $j$, the probability 
    that each bidder's value on item $j$ is smaller {than their} reserve price on item $j$, $r_{ij}$, is at least $1-b$. By the second condition, we know that, for any bidder $i$, the probability of {their} value for every item $j$ is below the reserve price $r_{ij}$ is at least $\half$. Moreover, using a similar argument as in Lemma~ \ref{lem:reserveprice}, we can show that at any equilibrium, any bidder whose value on an item is smaller than its reserve price will take the null action $\perp$ on that item.
    
    Consider bidder $i$ with type $t_i$ satisfying two conditions (i) $V_i(t_{ij}) \ge r_{ij}$, (ii) $\forall k \ne j, V_i(t_{ij}) < r_{ik}$. Then $i$ must bid $\perp$ on items other than $j$. Thus bidding $r_{ij}$ on item $j$ will lead to a non-negative utility $V_i(t_{ij}) - r_{ij}$ which is better than bidding $\perp$ on item $j$.

    We introduce the third condition (iii) $\forall i'\ne i, V_{i'}(t_{i'j}) < r_{i'j}$. Given that both conditions (ii) and (iii) are satisfied, as discussed in the preceding paragraph, bidder \( i \) will bid at least \( r_{ij} \) on item \( j \) whenever their value of item \( j \) is not less than the reserve price \( r_{ij} \), and will subsequently secure item \( j \).  Hence, the expected revenue from bidder $i$'s payment on item $j$ is at least $\E_{r_{ij}}\left[r_{ij}\cdot \Pr[V_i(t_{ij})\ge r_{ij}]\right]$. Since (ii) and (iii) are independent events, the joint probability of both conditions being satisfied is at least 
$\frac{1-b}{2}$. Consequently, the expected revenue generated by the randomized mechanism $\auc^{(r)}_{\mathrm{RP}}$ is at least $\frac{1-b}{2}\cdot \sum_{i,j}\E_{r_{ij}}\left[r_{ij}\cdot \Pr[V_i(t_{ij})\ge r_{ij}]\right]$.
    
\end{proof}

Since for any $j\in [m]$, $\sum_i\Pr\InBrackets{V_i(t_{ij})\ge p_{ij}} = \sum_i x_{ij}\cdot (q_i'/2) + y_{ij}\cdot (q_i''/2) = \sum_i q_{ij}/2 \le 1/2$ and for any $i\in [n]$, $\sum_j\Pr\InBrackets{V_i(t_{ij})\ge p_{ij}} = \sum_i x_{ij}\cdot (q_i'/2) + y_{ij}\cdot (q_i''/2) = \sum_j q_{ij}/2 \le 1/2$, we can apply lemma~ \ref{lem:randomreserveprice} to show $\sum_{i,j} \E_{p_{ij}}\InBrackets{p_{ij}\cdot \Pr\InBrackets{V_{i}(t_{ij})\ge p_{ij}}} \le 4\cdot \rev\InParentheses{\auc^{(p)}_{\mathrm{RP}}}$. Combining this with \eqref{eq:optud} and $\rev\InParentheses{\auc^{(p)}_{\mathrm{RP}}} \le \rprev$, we have proved the statement of our lemma.
\end{prevproof}

\subsection{Proof of Lemma~\ref{lem:approxtail}}
\label{subsec:tailappendix}
\begin{prevproof}{Lemma}{lem:approxtail}
    Let 
\[P_{ij} \in \text{argmax}_{x\ge c_i} \InParentheses {x + \beta_{ij}}\cdot \Pr_{\type_{ij}}\InBrackets{V_i\InParentheses{\type_{ij}} - \beta_{ij}\ge x},\]
and define 
\[r_{ij} := \InParentheses{P_{ij}+\beta_{ij}}\cdot \Pr_{\type_i}\InBrackets{V_i\InParentheses{\type_{ij}} - \beta_{ij}\ge P_{ij}} = \max_{x\ge c_i}\ \InParentheses{x+\beta_{ij}}\cdot \Pr_{\type_i}\InBrackets{V_i\InParentheses{\type_{ij}} -\beta_{ij}\ge x}, \]
$r_i = \sum_{j}r_{ij}$, and $r = \sum_i r_i$. We below show that $\tail(\beta)$ is upper bounded by $r$.

\begin{align}
    \begin{split}  
    \label{eq:tail}
    \tail(\beta) \le & \sum_i\sum_j\sum_{t_{ij}:V_i(t_{ij})\ge \beta+c_i}  f_i(\type_{ij}) \cdot \InParentheses{\beta_{ij} + c_i} \cdot \sum_{k\ne j}\Pr_{\type_{ik}}\InBrackets{V_i(\type_{ik}) - \beta_{ik} \ge V_i(\type_{ij}) - \beta_{ij}}\\
    & + \sum_i\sum_j\sum_{t_{ij}:V_i(t_{ij})\ge \beta+c_i}  f_i(\type_{ij}) \InParentheses{V_i(\type_{ij}) - \beta_{ij}} \cdot \sum_{k\ne j}\Pr_{\type_{ik}}\InBrackets{V_i(\type_{ik}) - \beta_{ik} \ge V_i(\type_{ij}) - \beta_{ij}}\\
    \le & \half \cdot \sum_i\sum_j\sum_{t_{ij}:V_i(t_{ij})\ge \beta+c_i}  f_i(\type_{ij}) \cdot \InParentheses{\beta_{ij} + c_i}
     + \sum_i\sum_j\sum_{t_{ij}:V_i(t_{ij})\ge \beta+c_i}  f_i(\type_{ij}) \cdot \sum_{k\ne j}r_{ik}\\
    \le & \half\sum_i\sum_j \Pr_{\type_{ij}}\InBrackets{V_i(\type_{ij})\ge \beta_{ij} + c_i} \cdot \InParentheses{\beta_{ij} + c_i} + \sum_i r_i\cdot \sum_j \Pr_{\type_{ij}}\InBrackets{V_i(\type_{ij})\ge \beta_{ij} + c_i}\\
    \le & \half\cdot \sum_i\sum_j r_{ij} + \half\cdot \sum_i r_i\\
    = & r.
    \end{split} 
\end{align}
In the second inequality, the first term is by $V_i(\type_{ij}) - \beta_{ij}\ge c_i$, so 
\[\sum_{k\ne j}\Pr_{\type_{ik}}\InBrackets{V_i(\type_{ik}) - \beta_{ik} \ge V_i(\type_{ij}) - \beta_{ij}} \le \sum_{k\ne j}\Pr_{\type_{ik}}\InBrackets{V_i(\type_{ik}) - \beta_{ik} \ge c_i}\le \half.\]
The second term is because $V_i(\type_{ij}) - \beta_{ij}\ge c_i$ and definiton of $r_{ik}$,
\[\InParentheses{V_i(\type_{ij}) - \beta_{ij}} \cdot \sum_{k\ne j}\Pr_{\type_{ik}}\InBrackets{V_i(\type_{ik}) - \beta_{ik} \ge V_i(\type_{ij}) - \beta_{ij}} \le \InParentheses{\beta_{ik} + V_i(\type_{ij}) - \beta_{ij}} \cdot \sum_{k\ne j}\Pr_{\type_{ik}}\InBrackets{V_i(\type_{ik}) - \beta_{ik} \ge V_i(\type_{ij}) - \beta_{ij}} \le r_{ik}.\]
As $P_{ij}\ge c_i$ and definition of $c_i$, $\left\{\beta_{ij}+P_{ij}\right\}_{i\in[n], j\in[m]}$ satisfies the conditions in lemma~\ref{lem:reserveprice}. Thus,
\[r = \sum_{i}\sum_{j}\InParentheses{\beta_{ij} + P_{ij}} \cdot \Pr_{\type_i}\InBrackets{V_i\InParentheses{\type_{ij}} \ge \beta_{ij} + P_{ij}} \le \frac{2}{1-b} \cdot \rprev.\]
Then our statement follows from this inequality and \eqref{eq:tail}.
\end{prevproof}

\section{Approximate Revenue Monotonicity}
\label{appendix:revmono}

\begin{theorem}
\label{thm:revmono}
Let \(\{v_i\}_{i \in [n]}\) be a set of valuation functions satisfying the properties of monotonicity, subadditivity, and no externalities. Consider two distributions, denoted by \(\dist= \bigtimes_{i\in[n]}\dist_i= \bigtimes_{i\in[n],j\in[m]}\dist_{ij}\) and \(\dist'= \bigtimes_{i\in[n]}\dist'_i= \bigtimes_{i\in[n],j\in[m]}\dist'_{ij}\), such that for each \(i\), distribution \(D_i'\) stochastically dominates distribution \(D_i\) with respect to valuation function \(v_i\). Specifically, there exists a coupling \((t_i, t_i')\) such that: (i) \(v_i(t_i,S) \leq v_i(t_i',S)\) for all \(S\subseteq [m]\), and (ii) the marginal distributions over \(t_i\) and \(t'_i\) correspond to \(D_i\) and \(D_i'\), respectively. Then, the following inequality holds:
\[\OPT(D') \geq \frac{1}{{229}
} \cdot \OPT(D).\]
\end{theorem}

\begin{proof}
We define $\prev$ as follows for distribution $F=\bigtimes_{i\in[n]} F_i=\bigtimes_{i\in[n], j\in[m]} F_{ij}$,
\[\prev(F) := \max_{b}\max_{r\in R(b)}\frac{1-b}{2}\sum_{i,j}\E_{r_{ij}}\InBrackets{r_{ij}\cdot \Pr_{t_{ij}\sim F_{ij}} \InBrackets{V_i(t_{ij})\ge r_{ij}}},\]
where $r=\{r_{ij}\}_{i\in[n],j\in[m]}$ and we use $R(b)$ to denote the set of reserve prices (possibly random) $r_{ij}$'s that satisfies the two conditions in Lemma~\ref{lem:randomreserveprice}, i.e., (1) $\sum_{i\in [n]}\Pr_{r_{ij}, t_{ij}\sim F_{ij}}[V_{i}(t_{ij})\ge r_{ij}]\le b$, $\forall j\in [m]$; (2) $\sum_{j\in[m]}\Pr_{r_{ij}, t_{ij}\sim F_{ij}}[V_{i}(t_{ij})\ge r_{ij}]\le \half$, $\forall i\in [n]$.

An easy fact is that $\prev(D')\ge \prev(D)$, because for any $i\in [n]$, $j\in [m]$ and $r_{ij}\ge 0$, there exists $r_{ij}'\ge 0$ such that $\Pr_{t_{ij}\sim D'_{ij}} \InBrackets{V_i(t_{ij})\ge r_{ij}'} = \Pr_{t_{ij}\sim D_{ij}} \InBrackets{V_i(t_{ij})\ge r_{ij}}$, and $r_{ij}'$ is greater than $r_{ij}$ as $D_{ij}$ is stochastically dominated by $D_{ij}'$.

By Lemma~\ref{lem:RevUpper}, and the proof of Lemma~\ref{lem:approxsingle}, Lemma~\ref{lem:approxtail} and Lemma~\ref{lem:coregap}, we know for any $b$,
\begin{equation}
\label{eq:optbound}
    \OPT(D) \le 4\cdot \coreh + \left(\frac{16b+8}{b(1-b)} + 16\right) \cdot \prev(D),
\end{equation}

where 
    \[\coreh = \sum_i\sum_{t_i\in T_i}\sum_{S\subseteq [m]}f_i(t_i)\sigma_{iS}(t_i)\val_i\InParentheses{\type_i,S\cap Y_i(\type_i)}. \] Here $f_i$ is the density function of $D_i$, and $Y_i(t_i) = \{j: V_i(t_{ij}) < \tau_i\}$, where $\{\tau_i\}_{i\in [n]}$ satisfies that $\sum_i \tau_i\le \frac{4}{1-b}\cdot \prev(D)$.

Let $s'$ be a Bayes-Nash equilibrium of simultaneous first price auction $\SFA$ w.r.t. type distribution $D'$ and valuation functions $\{v_i\}_{i\in[n]}$. Following Definition~\ref{def:good}, we define $\uti_i^{(s')}(t_i, S)$ to be the optimal interim utility of bidder $i$ with type $t_i$, when (a) all other bidders with type distributions $D_{-i}'$ bid according to $s_{-i}'$ and (b) they can only participate in the competition for items in $S$. Formally, 

\[\uti_i^{(s')}\InParentheses{t_i, S} = \sup_{q_i}\ \E_{\substack{\type_{-i}' \sim \dist_{-i}'\\ \bid_{-i}\sim \str_{-i}'\InParentheses{\type_{-i}'}}}\InBrackets{{\val_i\InParentheses{t_i,\alloc_i\InParentheses{q_i, \bid_{-i}} \cap S} -  {\sum_{j\in S}\pay_i^{(j)}\InParentheses{q^{(j)}_i, \bid^{(j)}_{-i}}}}}.\]

Now, by Lemma~\ref{lem:firstprice}, $(\SFA, {s'}, \dist, \{v_i\}_{i\in[n]})$ is $\frac12$-efficient, which means
\[\uti_i^{(s')}(t_i,S) + \rev^{(s')}_{D'}\InParentheses{\SFA, S}\ge \frac12 \cdot v_i(t_i, S).\]
By Lemma~\ref{lem:subadditiveofmu}, we know that $\utih_i^{(s')}(t_i,S) = \uti_i^{(s')}(t_i, S\cap Y_i(t_i))$ satisfies monotonicity, subadditivity, and no externalities. Similar to the proof of Lemma~\ref{lem:utih}, we can lower bound $\sum_i \E_{\type_i \sim \dist_i}\InBrackets{\utih_i^{(s')}\InParentheses{\type_i,[m]}}$, 
\begin{align*}
    \sum_i \E_{\type_i \sim \dist_i}\InBrackets{\utih_i^{(s')}\InParentheses{\type_i,[m]}} & \ge \sum_i \E_{\type_i \sim \dist_i} \InBrackets{\sum_{S\subseteq [m]} \sigma_{iS} \InParentheses{\type_i}\uti_i^{(s')}\InParentheses{\type_i,S\cap Y_i \InParentheses{\type_i}}}\\
    & \ge \sum_i \E_{\type_i \sim \dist_i}\InParentheses{\sum_{S\subseteq [m]} \sigma_{iS}\InParentheses{\type_i} \InParentheses{\half\cdot \val_i\InParentheses{\type_i, S\cap Y_i\InParentheses{\type_i}} - {\rev^{(s')}_{D'}\InParentheses{\SFA,S\cap Y_i\InParentheses{\type_i}}}}}\\
    & \ge \half\cdot \sum_i\sum_{\type_i \in T_i}\sum_{S\subseteq [m]} f_i\InParentheses{\type_i}\sigma_{iS}\InParentheses{\type_i}\val_i\InParentheses{\type_i, S\cap Y_i\InParentheses{\type_i}} - \sum_i\sum_{\type_i \in T_i}\sum_{S\subseteq [m]}f_i\InParentheses{\type_i}\sigma_{iS}\InParentheses{\type_i}\rev_{D'}^{(s')}\InParentheses{{\SFA,}S}\\
    & = \half\cdot \coreh - \sum_j \rev_{D'}^{(s')}\InParentheses{\SFA,\left\{j\right\}}\sum_i\sum_{t_i}f_i\InParentheses{\type_i}\sum_{S:j\in S}\sigma_{iS}\InParentheses{S}\\
    & = \half\cdot \coreh - \sum_j \rev_{D'}^{(s')}\InParentheses{\SFA,\left\{j\right\}}\sum_i\sum_{t_i}f_i\InParentheses{\type_i}\pi_{ij}\InParentheses{\type_i}\\
    & \ge \half\cdot \coreh - \rev_{D'}^{(s')}\InParentheses{\SFA, [m]}.
\end{align*}

And similar to the proof of Lemma~\ref{lem:entryfee}, let $e_i$ be the median of $\utih_i^{(s')}(t_i, [m])$ when $t_i$ is sampled from $D_i$. Since $\utih_i\InParentheses{\cdot, \cdot}$ is subadditive over independent items and  $\tau_i$-Lipschitz, we could apply Lemma~\ref{lem:concentrate} to get
\[\E_{\type_i \sim \dist_i}\InBrackets{\utih_i^{{(s')}}\InParentheses{\type_i, [m]}}\le 2e_i + \frac52 \cdot \tau_i.\]

Now consider drawing a sample $(t_i, t_i')$ from the joint distribution as described in the statement. Since $\val_i(t_i', S) \ge \val_i(t_i, S)$ for all $S\subseteq [m]$, the interim utility of bidder $i$ with type $t_i'$ is greater than $\uti_i^{(s')}(t_i, [m])$. And monotonicity of $\uti_i^{(s')}$ implies that $\uti_i^{(s')}\InParentheses{t_i, [m]} \ge \utih_i^{(s')}\InParentheses{t_i, [m]}$. Therefore, the interim utility of bidder $i$ with type $t_i'$ where $t_i'$ is sampled from $D_i'$ stochastically dominates the $\utih_i^{(s')}\InParentheses{t_i, [m]}$ where $t_i$ is from $D_i$. Thus, if we set the entry fee as $e_i$, i.e.,  the median of $\utih_i^{(s')}(t_i, [m])$, the probability that bidder $i$ from distribution $D_i'$ pays the entry fee is at least $1/2$. Thus
\[\efrev^{(s')}_{D'}\InParentheses{\SFA} \geq \sum_i e_i\Pr_{\type_i' \sim \dist_i'}\left[\uti_i^{(s')}\InParentheses{\type_i', [m]} \ge e_i\right] \ge \half\cdot \sum_i e_i.\]

Combining the two inequalities above, we know
\begin{align*}
    \sum_i \E_{\type_i \sim \dist_i}\InBrackets{\utih_i^{(s')}\InParentheses{\type_i,[m]}} & \le 2\sum_i e_i + \frac52\sum_i\tau_i\\
    & \le 4\cdot\efrev^{(s')}_{D'}\InParentheses{\SFA} + \frac52\sum_i\tau_i.
\end{align*}

By the obtained lower and upper bound of $\sum_i \E_{\type_i \sim \dist_i}\InBrackets{\utih_i^{(s')}\InParentheses{\type_i,[m]}}$, we have
\begin{align*}
    \coreh & \le 8\cdot\efrev^{(s')}_{D'}\InParentheses{\SFA} + 2\cdot \rev_{D'}^{(s')}\InParentheses{\SFA} + 5\cdot \sum_i\tau_i\\
    & \le 8\cdot\efrev^{(s')}_{D'}\InParentheses{\SFA} + 2\cdot \rev_{D'}^{(s')}\InParentheses{\SFA} + \frac{20}{1-b}\cdot \prev(D).
\end{align*}

Plugging this into \eqref{eq:optbound}, and taking $b = \frac15$,
\begin{align*}
    \OPT(D) & \le 32\cdot\efrev^{(s')}_{D'}\InParentheses{\SFA} + 8\cdot \rev_{D'}^{(s')}\InParentheses{\SFA} + 186\cdot \prev(D)\\
    & \le 32\cdot\efrev^{(s')}_{D'}\InParentheses{\SFA} + 8\cdot \rev_{D'}^{(s')}\InParentheses{\SFA} + 186\cdot \rprev(D')\\
    & \le 42\cdot \rev^{(s')}_{\dist'}\InParentheses{\mathrm{S1A}^{(e)}_{\mathrm{EF}}} + {187\cdot \OPT(D')}\\
    & \le {229}\cdot \OPT(D').
\end{align*}
The second inequality is due to $\prev(D)\le \prev(D')$ and $\prev(D')\le \rprev(D')$ by Lemma~\ref{lem:randomreserveprice}. The third inequality is by lemma~\ref{lem:entryrev}.

\end{proof}
 
\end{document}